\documentclass{article}

\usepackage{amsmath}
\usepackage{amssymb}
\usepackage{float}
\usepackage{natbib}
\usepackage[utf8]{inputenc}
\usepackage{enumerate}
\usepackage{rotating}
\usepackage{graphicx,psfrag,epsf}
\usepackage{amsthm}

\usepackage[top=2.5cm, bottom=2cm, left=1.75cm, right=1.75cm]{geometry}
\usepackage{tikz}
\usetikzlibrary{shapes,arrows}
\usepackage[none]{hyphenat} 

\newtheorem{proposition}{Proposition}

\title{The Raise Regression: Justification, properties and application}
\author{Rom\'an Salmer\'on G\'omez \\ %
    Departamento de M\'etodos Cuantitativos para la Econom\'ia de la Empresa \\ %
    Universidad de Granada \\ %
    Catalina Garc\'ia Garc\'ia \\ %
    Departamento de M\'etodos Cuantitativos para la Econom\'ia de la Empresa \\ %
    Universidad de Granada \\ %
    Jos\'e Garc\'ia P\'erez \\ %
    Departamento de Administraci\'on y Direcci\'on de Empresas \\ %
    Universidad de Almer\'ia %
}

\begin{document}

    \sloppy 

    \maketitle

    \begin{abstract}
Multicollinearity produces an inflation in the variance of the Ordinary Least Squares estimators due to the correlation between two or more independent variables (including the constant term). A widely applied solution is to estimate with penalized estimators (such as the ridge estimator, the Liu estimator, etc.) which exchange the mean square error by the bias. Although the variance diminishes with these procedures, all seems to indicate that the inference is lost and also the goodness of fit. Alternatively, the raise regression (\cite{Garcia2011} and \cite{Salmeron2017}) allows the mitigation of the problems generated by multicollinearity but without losing the inference and keeping the coefficient of determination. This paper completely formalizes the raise estimator summarizing all the previous contributions: its mean square error, the variance inflation factor, the condition number, the adequate selection of the variable to be raised, the successive raising and the relation between the raise and the ridge estimator. As a novelty, it is also presented the estimation method, the relation between the raise and the residualization, it is analyzed the norm of the estimator and the behaviour of the individual and joint significance test and the behaviour of the mean square error and the coefficient of variation. The usefulness of the raise regression as alternative to mitigate the multicollinearity is illustrated with two empirical applications.
    \end{abstract}


\section{Introduction}

The presence of collinearity between independent variables in a regression models leads to serious problems with the ordinary least squares (OLS) estimator which becomes unstable. According to \cite{willan1978meaningful}, the multicollinearity produces ``inflated variances and covariances, inflated correlations, inflated prediction variances and the concomitant difficulties in interpreting the significance values and confidence regions for parameter''.

The ridge estimation is an early answer to the problem of estimating a model with collinearity. The solution of the normal equations is an algebraic problem within the numerical analysis and, for this reason, the beginnings of ridge estimator could be found in the numerical analysis. Following \cite{om2001ridge} and \cite{singh2011}, the first antecedent is found in the Tikhonov regularization (TR), which is one of the most widely applied techniques to solve the problem of a ill-conditioning matrix. \cite{om2001ridge} concluded that the TR applied in numerical analysis has been presented in statistics and econometrics as ridge regression (RR) and, in some way, he claimed for Tikhonov the merit of the ridge regression due to he presented its regularization in 1943 while ridge regression was formally presented by \cite{HoerlKennard1970a} and \cite{HoerlKennard1970b}.

However, one question is to solve a system of equations with an ill-conditioning matrix and another question is to find an estimator for a vector of parameters in a lineal model. It is clear than an estimator is an expression that can be obtained from resolving a system of equations but an estimator must verify also some characteristics as being a random variable and allow the obtention of inference from it. Thus, while the numerical analysis focus on the resolution of a system of equation, the statistic and the econometric are focus on the qualities and properties of the estimator since an estimator should be accompanied by the possibility of doing inference and prediction which are the main reasons for its estimation.

Several authors consider that the origin of ridge regression is found on the paper presented by \cite{levenberg1944method} where the increase of matrix $\mathbf{X}^{t} \mathbf{X}$ are presented for the first time, see \cite{piegorsch1989early}.  \cite{marquardt1963algorithm} developed a procedure to Levenberg but without express knowledge of it. \cite{riley1955solving} is the first that applied the expression $\mathbf{C} = \mathbf{A} + k \mathbf{I}$ and showed that the matrix $\mathbf{C}$ is better conditioned than the matrix $\mathbf{A}$. \cite{marquardt1963algorithm} and Hoerl (1962) showed that to invert matrix $\mathbf{X}^{t} \mathbf{X} + k \mathbf{I}$ is easier than to invert matrix $\mathbf{X}^{t} \mathbf{X}$. Despite of this fact, the most of the references attribute the ridge estimator to \cite{HoerlKennard1970a} and \cite{HoerlKennard1970b}, who proposed the ridge estimator for a linear model $\mathbf{Y}_{n \times 1} = \mathbf{X}_{n \times p} \boldsymbol{\beta}_{p \times 1} + \mathbf{u}_{n \times 1}$  where $\mathbf{u} \rightsquigarrow N(0,\sigma^2 \mathbf{I})$ with the following expression:
\begin{equation}
  \widehat{\boldsymbol{\beta}}_{R}(k) = (\mathbf{X}^{t} \mathbf{X} + k \mathbf{I})^{-1} \mathbf{X}^{t} \mathbf{Y},
  \end{equation}
being a biased estimator that coincides to ordinary least squares (OLS) when $k=0$. These papers formalized the ridge estimator and showed that it verifies the mean squared error (MSE) admissibility condition assuring an improvement in MSE to OLS  for some $k\in(0,\infty)$. Thus, this estimator provided a solution to the problems indicated by \cite{willan1978meaningful}.

Regardless of the (always post-hoc) theoretical justification for the introduction of the values of $k$ in the main diagonal of matrix $\mathbf{X}^{t} \mathbf{X}$, the ridge estimator and, in general, the constrained estimators present problems with the measure of goodness of fit and with the inference. These problems appear due to the introduction of the value $k$ (a non deterministic value) which impedes the verification of the traditional sum of squares decomposition (see, for more detail, \cite{Salmeron2020}).

In relation to this problem, \cite{hoerl1990ridge} indicated: ``It appears  to be common practice with ridge regression to obtain a decomposition of the total sum of squares, and assign degrees of freedom according to established  least squared theory''  emphasising the mistake of this practice. These authors presented a descomposition for the variance analysis concluding that ``In any case, \cite{obenchain1977classical} has shown that is no unique exact F tests possible with non stochastic ridge estimators''. Formerly, \cite{schmidt1976econometrics} (p. 53) indicated that the ridge regression (RR) may not be useful to econometricians because the underdeveloped theory of hypothesis testing with RR limits its utility. He may be the first author who questioned the inference in ridge regression.

It seems to be appropriate to ponder the wide success of the ridge estimator endorsed by numerous references in the scientific literature in empirical and theoretical fields. Certainly, the ridge regression allows to overpass the problems indicated by  \cite{willan1978meaningful}, but in exchange for this, is the inference lost?
In this line, it is interesting to comment the paper by \cite{smith1980critique} entitled ``A Critique of Some Ridge Regression Methods'' which was included together to the comments and the final answer in the paper 496 of Cowles Fundation Papers. This document collects different criticism to the ridge estimator in relation to its bayesian version, the lack of invariance to scale changes and the problems with inference on which we focus now.


The first paper that addresses in depth the problem of inference is the one of \cite{obenchain1977classical} who provided a rigorous analysis of hypothesis testing and confidence interval construction in a (generalized) ridge regression when the shrinkage factor is deterministic. The abstract of this paper indicates ``testing general linear hypotheses in multiple regression models, it is shown that non-stochastically shrunken ridge estimators yield the same central F-ratios and t-statistics as does the least squares estimator''. From this statement we raise the following question:  What is the relation between these results with the inference in OLS?. On the other hand, we also consider that the assumption that the shrinkage is non-stochastic excessively limits the applicability of the results.

By following \cite{obenchain1977classical}, other authors such as \cite{coutsourides1979f} insisted in the results of \cite{obenchain1977classical} and extend them (with the same restrictions) to the RR, principal components (PC) and Shrunken Estimators (SH) yield the same central t and F statistics as the OLS estimator. They concluded that ``It is important to point out that confidence intervals which are centred at biased estimators have greater maximum diameter that the least squares interval of the same confidence''. We raise another question: if the interval for the biased estimator have a greater diameter, then the tendency (in presence of multicollinearity) to not reject the null hypothesis of individual significance test ($\beta_i=0$) while the null hypothesis of the joint significance test is rejected, will be even higher and, then,  the application of the ridge estimation makes not statistical sense due to we only have stabilized the estimation but we have not obtained a model which overpass the validation step. On the other hand, the restrictions about the parameter $k$ will be the same: it will be not stochastic as stated by \cite{thisted1980critique} in the comment to the paper of  \cite{smith1980critique}, when indicated that the authors ``completely disregarded some of the most promising aspect of the ridge regression such as the adaptive estimation of $k$''. \cite{stahlecker1996biased} also extended the work presented by \cite{obenchain1977classical} to a more wide range of estimators that content the ridge estimator as a particular case, but keeping the deterministic hypothesis of $k$. \cite{hoerl1990ridge} had already indicated the problem of the sum of squares descomposition which is common to all the penalized estimator such as the Lasso and the elastic net apart from the ridge regression. \cite{vinod1987confidence} dedicated a full chapter to present the confidence intervals for the ridge estimator concluding that, if the economist wishes to use a confidence interval focused on the ridge estimator, three possibilities exist: the approximate Bayes, bootstrapping, those based on \cite{stein1981estimation} unbiased estimate of the mean squared error (MSE) of a biased estimator of multivariate normal mean and the bootstrap. \cite{van1980critique}, in other comment to the paper of \cite{smith1980critique}, goes even further and points in relation to inference that ``It must be  emphasized that these condition  are commonly not met; in particular, $k$ is usually stochastic  and not fixed'' supporting this statement in the papers of \cite{mcdonald1975monte}, \cite{bingham1977simulation} and in its own analysis of the study of \cite{lawless1978ridge}, that according to  \cite{van1980critique} seem to indicate it roughly holds for stochastically chosen  $k$. Also see \cite{thisted1976ridge}.

The inference of the ridge regression is still a debatable and controversial question. Thus, \cite{halawa2000tests} indicated that there are numerous works in relation to the obtention of the ridge estimator, but few works focused on the individual or joint inference of the parameters, citing among others \cite{obenchain1977classical}, \cite{coutsourides1979f}, \cite{ullah1984sampling}, \cite{ohtani1985bounds} and cited the sum of squares decomposition of \cite{hoerl1990ridge} and the degree of freedom for the analysis of the variance. \cite{halawa2000tests} investigated two non-exact,  ridge-based, type tests, for a certain values of  $k$ checked with simulations.  \cite{monarrez2007regresion} says to establish the conditions required to get that the central t-Student distribution applied in OLS will be also applicable in ridge regression but it may receive the same criticism and limitations than previous.   \cite{bae2014general} derive a test statistic for the general linear test in the RR model. The exact distribution for the test statistic is too difficult to derive; therefore, they suggest an approximate reference distribution and used numerical studies to verify that the suggested distribution for the test statistic is appropriate. \cite{gokpinar2016study}, by following  \cite{halawa2000tests} investigate others popular $k$ values used the RR for testing significance of regression coefficient and compare tests in terms of type I error rates and powers by using Monte Carlo simulation. \cite{sengupta2020asymptotic} presented another recent attempt within the field of the inference in ridge regression characterizing the asymptotic distribution of the ridge estimation when the adjustment parameter is selected with a test sample. Finally, we want to consider the work of \cite{vanhove2020collinearity} entitled ``Collinearity isn’t a disease that needs curing'',  where  the lack of inference in ridge estimator is considered as an argument against the treatment of the multicollinearity, based on the recent paper presented by \cite{goeman2018l1}.

Certainly, the multicollinearity produces an inflation in the variance due to the correlation between two or more independent variables (constant term included), but when the problem is treated applying penalized estimator, the mean squared error is exchanged by biased. Additionally, although the variance is diminished, all seems to indicate that the inference is lost  and, not only that, but also the goodness of fit is lost since the coefficient of determination makes no sense at least in its traditional form.

It is difficult to understand how being the inference of ridge regression questioned and despite of the numerous paper published in relation to ridge estimator, only a few of them are focused on the inference. Indeed, there are not much more works in relation to inference in ridge regression that the ones cited in this paper that all found its origin in the work of \cite{obenchain1977classical}. It is possible that the problem of inference in ridge regression has no solution and this fact motivate the formalization of the raise regression as a possible alternative.

The raise estimator allows, as the penalized estimator, mitigate the problems indicated by \cite{willan1978meaningful} but without losing the inference and keeping the coefficient of determination, not only the same value of the OLS regression but also its statistical interpretation due to the sum of squares is verified after the application of the raise estimator. In addition, the basis in which the raise estimator is founded lead, in a natural way, to the inclusion of the values of $k$ in the main diagonal of matrix $\mathbf{X}^{t} \mathbf{X}$ without being necessary any other justification. Indeed, recently \cite{Garcia2020} showed that the ridge estimator is a particular case of the raise estimator or the raising procedures where this estimator founds its origin. Ultimately, the goal of this paper is to present a global vision of the raising procedure, making new contributions to this technique and summarizing those already carried out in previous works (in order to have a better understanding of this technique). In addition, the procedure is performed for any number of regressors and establishing relationships with OLS and residualization, the latter technique also applied when there is a worrying degree of multicollinearity (see \cite{Garcia2019} for more details).

The structure of the paper is as follows:  Section \ref{RaiseRegression} presents as a novelty the generalization of its estimation and the relation between the raise regression with the ordinary least squares and with residualization. It is also presented originally its goodness of fit, global characteristics and individual inference. This section also includes a review about previous results in relation to the mean square error of the raise regression.
Section \ref{multicol} presents as a novelty the norm of the raised estimator, the variance of the estimators and it is shown that the raise regression is able to mitigate not only the essential multicollinearity but also the non-essential one (due to its coefficient of variation increases when the raise regression is applied)
together to a summary of how to obtain the Variance Inflation Factors and the Condition Number.
Section \ref{ChoiceVarLanda} proposes some criteria, some of them also presented in an original way, to select the variable to be raised that can be applied individually or as a combination.
Due to the Variance Inflation Factor and the Condition Number presents lower bound in raise regression, it is possible that may not exist a raising factor that get the enough mitigation of the multicollinearity. In this case, it is possible to apply a successive raise regression that is summarized in Section \ref{SucRaise}.
Section \ref{MtxMSE} compares, also as a novelty, the OLS, the raise and the ridge estimators under the criteria of MSE and also summarizes the contribution of \cite{Garcia2020} who conclude that the ridge estimator is a particular case of the raising procedures.
The usefulness of the raise estimator is illustrated with two empirical applications in Section \ref{example}.
Finally, Section \ref{conclusion} summarizes the main contributions of this paper.

\section{Raise regression}
    \label{RaiseRegression}

    This section analyze the raise regression presenting as a novelty, for any value of $p$, its estimation, goodness of fit, global characteristics and individual inference. In all cases, it is also analyze its relation with the OLS estimator of the initial model and with the residualization (see, for more details about this technique, \cite{Garcia2019}). Due to the estimators obtained with this methodology are biased, it will be interesting to analyzed its Mean Square Error (MSE). Thus, it will be also summarized the results showed by \cite{Salmeron2020} in relation to MSE.

    \subsection{Estimation}

    Considering the multiple linear regression model for $p$ independent variables and $n$ observations:
    \begin{equation}
        \mathbf{Y} = \beta_{1} + \beta_{2} \mathbf{X}_{2} + \cdots + \beta_{i} \mathbf{X}_{i} + \cdots + \beta_{p} \mathbf{X}_{p} + \mathbf{u} = \mathbf{X} \boldsymbol{\beta} + \mathbf{u},
        \label{modelo0}
    \end{equation}
    being $\mathbf{Y}$ a vector $n \times 1$ containing the observations of the dependent variable, $\mathbf{X} = [\mathbf{1} \ \mathbf{X}_{2} \dots \mathbf{X}_{i} \dots \mathbf{X}_{p}] = [ \mathbf{X}_{i} \ \mathbf{X}_{-i}]$ is a matrix with order $n \times p$ that contains (by columns) the observations of the independent variables (where $\mathbf{1}$ is a vector of ones with dimension $n \times 1$ and $\mathbf{X}_{-i}$, with $i\geq2$, is the results of eliminating the column $i$ of matrix $\mathbf{X}$), $\boldsymbol{\beta}$ is a vector $p \times 1$ that contains the coefficient of the independent variables and $\mathbf{u}$ is a vector $n \times 1$ that represents the random disturbance, which is supposed to be spherical (it is to say, $E[\mathbf{u}] = \mathbf{0}$ and $Var(\mathbf{u}) = \sigma^{2} \mathbf{I}$ where $\mathbf{0}$ is a vector of zeros with dimension $n \times 1$ and $\mathbf{I}$ the identity matrix of adequate dimensions $p \times p$).

    Considering the residuals, $\mathbf{e}_{i}$, of the auxiliary regression of the dependent variable $\mathbf{X}_{i}$, with $i\geq2$, as a function of the rest of the independent variables of the model (\ref{modelo0}):
    \begin{equation}
        \mathbf{X}_{i} = \alpha_{1} + \alpha_{2} \mathbf{X}_{2} + \cdots + \alpha_{i-1} \mathbf{X}_{i-1} + \alpha_{i+1} \mathbf{X}_{i+1} + \cdots + \alpha_{p} \mathbf{X}_{p} + \mathbf{v} = \mathbf{X}_{-i} \boldsymbol{\alpha} + \mathbf{v}, \label{modelo_aux}
    \end{equation}
    the variable $i$ is raised as $\widetilde{\mathbf{X}}_{i} = \mathbf{X}_{i} + \lambda \cdot \mathbf{e}_{i}$ with $\lambda \geq 0$ (called raising factor) and verifying\footnote{
    In addition, due to $\mathbf{e}_{i}^{t} \mathbf{X}_{-i} = \mathbf{0}$ is verified that:
    \begin{eqnarray}
        \mathbf{X}_{-i}^{t} \widetilde{\mathbf{X}}_{i} &=& \mathbf{X}_{-i}^{t} \cdot \left( \mathbf{X}_{i} + \lambda \cdot \mathbf{e}_{i} \right) = \mathbf{X}_{-i}^{t} \mathbf{X}_{i}, \nonumber \\
        \widetilde{\mathbf{X}}_{i}^{t} \mathbf{Y} &=& \left( \mathbf{X}_{i} + \lambda \cdot \mathbf{e}_{i} \right)^{t} \mathbf{Y} = \mathbf{X}_{i}^{t} \mathbf{Y} + \lambda \cdot \mathbf{e}_{i}^{t} \mathbf{Y}, \nonumber \\
        \widetilde{\mathbf{X}}_{i}^{t} \widetilde{\mathbf{X}}_{i} &=& \left( \mathbf{X}_{i} + \lambda \cdot \mathbf{e}_{i} \right)^{t} \left( \mathbf{X}_{i} + \lambda \cdot \mathbf{e}_{i} \right) = \mathbf{X}_{i}^{t} \mathbf{X}_{i} + (\lambda^{2} + 2\lambda) \cdot \mathbf{e}_{i}^{t}\mathbf{e}_{i}, \nonumber
    \end{eqnarray}
    where was used that $\mathbf{e}_{i}^{t} \mathbf{X}_{i} = \mathbf{e}_{i}^{t} \cdot \left( \mathbf{X}_{-i} \widehat{\boldsymbol{\alpha}} + \mathbf{e}_{i}^{t} \right) = \mathbf{e}_{i}^{t} \mathbf{e}_{i}.$
    } that $\mathbf{e}_{i}^{t} \mathbf{X}_{-i} = \mathbf{0}$.

    In this case, the raise regression consists in the estimation by OLS of the following model:
    \begin{equation}
        \mathbf{Y} = \beta_{1}(\lambda) + \beta_{2}(\lambda) \mathbf{X}_{2} + \cdots + \beta_{i}(\lambda) \widetilde{\mathbf{X}}_{i} + \cdots + \beta_{p}(\lambda) \mathbf{X}_{p} + \widetilde{\mathbf{u}} = \widetilde{\mathbf{X}} \boldsymbol{\beta}(\lambda) + \widetilde{\mathbf{u}},
        \label{modelo_alzado}
    \end{equation}
    where $\widetilde{\mathbf{X}} = [\mathbf{1} \ \mathbf{X}_{2} \dots \widetilde{\mathbf{X}}_{i} \dots \mathbf{X}_{p}] = [\mathbf{X}_{-i} \  \widetilde{\mathbf{X}}_{i}]$ and $\widetilde{\mathbf{u}}$ is a random disturbance which is also supposed to be spherical.
    The OLS estimator of the model (\ref{modelo_alzado}) will be given by:
    \begin{eqnarray}
        \widehat{\boldsymbol{\beta}}(\lambda) &=& \left( \widetilde{\mathbf{X}}^{t} \widetilde{\mathbf{X}} \right)^{-1} \widetilde{\mathbf{X}}^{t} \mathbf{Y} = \left(
            \begin{array}{cc}
                \mathbf{X}_{-i}^{t} \mathbf{X}_{-i} & \mathbf{X}_{-i}^{t} \widetilde{\mathbf{X}}_{i} \\
                \widetilde{\mathbf{X}}_{i}^{t} \mathbf{X}_{-i} & \widetilde{\mathbf{X}}_{i}^{t} \widetilde{\mathbf{X}}_{i}
            \end{array} \right)^{-1} \cdot \left(
            \begin{array}{c}
                \mathbf{X}_{-i}^{t} \mathbf{Y} \\
                \widetilde{\mathbf{X}}_{-i}^{t} \mathbf{Y}
            \end{array} \right) \nonumber \\
            &=& \left(
            \begin{array}{cc}
                \mathbf{X}_{-i}^{t} \mathbf{X}_{-i} & \mathbf{X}_{-i}^{t} \mathbf{X}_{i} \\
                \mathbf{X}_{i}^{t} \mathbf{X}_{-i} & \mathbf{X}_{i}^{t} \mathbf{X}_{i} + (\lambda^{2} + 2\lambda) \cdot \mathbf{e}_{i}^{t}\mathbf{e}_{i}
            \end{array} \right)^{-1} \cdot \left(
            \begin{array}{c}
                \mathbf{X}_{-i}^{t} \mathbf{Y} \\
                \mathbf{X}_{-i}^{t} \mathbf{Y} + \lambda \mathbf{e}_{i}^{t} \mathbf{Y}
            \end{array} \right) \nonumber \\
            &\underbrace{=}_{Appendix \ \ref{appendixA}}& \left(
                \begin{array}{c}
                    \left( \mathbf{X}_{-i}^{t} \mathbf{X}_{-i} \right)^{-1} \mathbf{X}_{-i}^{t} \mathbf{Y} - \widehat{\boldsymbol{\alpha}} \cdot \frac{\mathbf{e}_{i}^{t} \mathbf{Y}}{(1+\lambda) \cdot \mathbf{e}_{i}^{t} \mathbf{e}_{i}} \\
                    \frac{\mathbf{e}_{i}^{t} \mathbf{Y}}{(1+\lambda) \cdot \mathbf{e}_{i}^{t} \mathbf{e}_{i}}
                \end{array} \right)  = \left(
                \begin{array}{c}
                    \widehat{\boldsymbol{\beta}}(\lambda)_{-i} \\
                    \widehat{\beta}(\lambda)_{i}
                \end{array} \right), \label{estimador_raise}
    \end{eqnarray}
    being $\mathbf{e}_{i}^{t} \mathbf{e}_{i}$ the sum of squares of residuals (SSR) of model (\ref{modelo_aux}), $\widehat{\boldsymbol{\beta}}(\lambda)_{-i}$ are the estimator of the coefficients associated to the non altered variables in  $\mathbf{X}_{-i}$ and $\widehat{\beta}(\lambda)_{i}$ is the coefficient of the raise variable $\widetilde{\mathbf{X}}_{i}$.

    Taking into account the Appendix \ref{appendix_ortogonal}, that summarizes the residualization formally developed in \cite{Garcia2019}, it is possible to conclude that:
    \begin{itemize}
        \item For $\lambda \rightarrow +\infty$, the estimation of the not altered variables coincide to the estimation obtained with orthogonal variables:
            $$\lim \limits_{\lambda \rightarrow +\infty} \widehat{\boldsymbol{\beta}}(\lambda)_{-i} = \left( \mathbf{X}_{-i}^{t} \mathbf{X}_{-i} \right)^{-1} \mathbf{X}_{-i}^{t} \mathbf{Y} = \boldsymbol{\gamma}_{-i}.$$
        \item  For $\lambda \rightarrow +\infty$, the estimation of the raise variable tends to zero: $\lim \limits_{\lambda \rightarrow +\infty} \widehat{\beta}(\lambda)_{i} = 0$.
    \end{itemize}

    Finally, taking into account that, for $i=2,\dots,p$, it is obtained that:
    $$\widetilde{\mathbf{X}}_{i} = \mathbf{X}_{i} + \lambda \cdot \mathbf{e}_{i} = (1 + \lambda) \cdot \mathbf{X}_{i} - \lambda \cdot \left( \widehat{\alpha}_{0} + \widehat{\alpha}_{1} \mathbf{X}_{1} + \cdots + \widehat{\alpha}_{i-1} \mathbf{X}_{i-1} + \widehat{\alpha}_{i+1} \mathbf{X}_{i+1} + \cdots + \widehat{\alpha}_{p} \mathbf{X}_{p} \right),$$
    and, consequently, the matrix $\widetilde{\mathbf{X}}$ of expression (\ref{modelo_alzado}) can be written as $\widetilde{\mathbf{X}} = \mathbf{X}\cdot \mathbf{M}_{\lambda}$ where:
    \begin{equation}
        \label{mlanda}
        \mathbf{M}_{\lambda} = \left(
        \begin{array}{cccccccc}
            1 & 0 & \cdots & 0 & - \lambda \widehat{\alpha}_{0} & 0 & \cdots & 0 \\
            0 & 1 & \cdots & 0 & - \lambda \widehat{\alpha}_{1} & 0 & \cdots & 0 \\
            \vdots & \vdots & & \vdots & \vdots & \vdots & & \vdots \\
            0 & 0 & \cdots & 1 & - \lambda \widehat{\alpha}_{i-1} & 0 & \cdots & 0 \\
            0 & 0 & \cdots & 0 & 1 + \lambda & 0 & \cdots & 0 \\
            0 & 0 & \cdots & 0 & - \lambda \widehat{\alpha}_{i+1} & 1 & \cdots & 0 \\
            \vdots & \vdots & & \vdots & \vdots & \vdots & & \vdots \\
            0 & 0 & \cdots & 0 & - \lambda \widehat{\alpha}_{p} & 0 & \cdots & 1 \\
        \end{array} \right).
    \end{equation}

    Thus, we have that:
    $$\widehat{\boldsymbol{\beta}}(\lambda)=(\widetilde{\mathbf{X}}^t\cdot\widetilde{\mathbf{X}})^{-1}\widetilde{\mathbf{X}}^t\cdot\mathbf{Y} = \mathbf{M}_{\lambda}^{-1} \cdot \widehat{\boldsymbol{\beta}},$$
    and then the estimator of $\boldsymbol{\beta}$ obtained from (\ref{modelo_alzado}) is biased unless $\mathbf{M}_{\lambda} = \mathbf{I}$, which only occurs when $\lambda=0$. It is to say, when the raise regression coincides to the ordinary least square regression.

    \subsection{Mean Square Error}
        \label{mse_raise}

    Since $\widehat{\boldsymbol{\beta}}(\lambda)$ is a biased estimator of $\boldsymbol{\beta}$ when $\lambda \not= 0$, it will be interesting to obtain its mean square error (MSE). From \cite{Salmeron2020} the MSE for raise regression is given by:
    \begin{equation}
        \text{MSE}\left(\widehat{\boldsymbol{\beta}}(\lambda)\right) = \sigma^{2} tr \left( \left( \mathbf{X}_{-i}^{t} \mathbf{X}_{-i} \right)^{-1} \right) + \left( 1 + \sum \limits_{j=0, j \not= i}^{p} \widehat{\alpha}_{j}^{2} \right) \cdot \beta_{i}^{2} \cdot \frac{\lambda^{2} + h}{(1+\lambda)^{2}},
    \end{equation}
    where $h = \frac{\sigma^{2}}{\beta_{i}^{2} \cdot \mathbf{e}_{i}^{t} \mathbf{e}_{i}}$ and $i=2,\dots,p$.

    Once it is obtained the estimations of $\sigma^2$ and $\beta_{i}$ from model (\ref{modelo0}) by OLS, $\lambda_{min} = \frac{\widehat{\sigma}^{2}}{\widehat{\beta}_{i}^{2} \cdot \mathbf{e}_{i}^{t} \mathbf{e}_{i}}$ minimizes MSE and it is verified that $\text{MSE}\left(\widehat{\boldsymbol{\beta}}(\lambda_{min})\right) < \text{MSE}\left(\widehat{\boldsymbol{\beta}}\right)$.
    That is, it is always possible to obtain a value for $\lambda$ that verifies that it has an associated MSE less than that of OLS.

    Also, if $\lambda_{min} > 1$ it is obtained that $\text{MSE}\left(\widehat{\boldsymbol{\beta}}(\lambda)\right) < \text{MSE}\left(\widehat{\boldsymbol{\beta}}\right)$ for all $\lambda > 0$. That is, it is always verified that the MSE of the raised is less than that of the OLS.
    If $\lambda_{min} < 1$, then $\text{MSE}\left(\widehat{\boldsymbol{\beta}}(\lambda)\right) \leq \text{MSE}\left(\widehat{\boldsymbol{\beta}}\right)$ for $\lambda \leq 2 \cdot \frac{\widehat{\sigma}^{2}}{\widehat{\beta}_{i}^{2} \cdot \mathbf{e}_{i}^{t} \mathbf{e}_{i} - \widehat{\sigma}^{2}}$.

    Finally:
    $$\lim \limits_{\lambda \rightarrow +\infty} \text{MSE}\left(\widehat{\boldsymbol{\beta}}(\lambda)\right) = \sigma^{2} tr \left( \left( \mathbf{X}_{-i}^{t} \mathbf{X}_{-i} \right)^{-1} \right) + \left( 1 + \sum \limits_{j=0, j \not= i}^{p} \widehat{\alpha}_{j}^{2} \right) \cdot \beta_{i}^{2},$$
    expression that differs from the MSE of residualization (see Appendix \ref{appendix_ortogonal}).

    \subsection{Goodness of fit and global characteristics}
        \label{goodness_raise}

    To analyze the goodness of fit of model (\ref{modelo_alzado}), its residuals are obtained from the following expression:
    \begin{eqnarray}
        \mathbf{e}(\lambda) &=& \mathbf{Y} - \widetilde{\mathbf{X}} \widehat{\boldsymbol{\beta}}(\lambda) = \mathbf{Y} - ( \mathbf{X}_{-i} \ \widetilde{\mathbf{X}}_{i} ) \cdot \left(
            \begin{array}{c}
                \widehat{\boldsymbol{\beta}}(\lambda)_{-i} \\
                    \widehat{\beta}(\lambda)_{i}
                \end{array} \right) \nonumber \\
            &=& \mathbf{Y} - \mathbf{X}_{-i} \left( \mathbf{X}_{-i}^{t} \mathbf{X}_{-i} \right)^{-1} \mathbf{X}_{-i }^{t} \mathbf{Y} + \left( (\lambda + 1) \cdot \mathbf{e}_{i}^{t} \mathbf{e}_{i} \right)^{-1} \mathbf{X}_{-i} \cdot \widehat{\boldsymbol{\alpha}} \cdot \mathbf{e}_{i}^{t} \mathbf{Y} \nonumber \\
            & & - \left( (\lambda + 1) \cdot \mathbf{e}_{i}^{t} \mathbf{e}_{i} \right)^{-1} \cdot \mathbf{X}_{i} \cdot \mathbf{e}_{i}^{t} \mathbf{Y} - \lambda \left( (\lambda + 1) \cdot \mathbf{e}_{i}^{t} \mathbf{e}_{i} \right)^{-1} \cdot \mathbf{e}_{i} \cdot \mathbf{e}_{i}^{t} \mathbf{Y} \nonumber \\
            &=& \mathbf{Y} - \mathbf{X}_{-i} \left( \mathbf{X}_{-i}^{t} \mathbf{X}_{-i} \right)^{-1} \mathbf{X}_{-i}^{t} \mathbf{Y} + \left( (\lambda + 1) \cdot \mathbf{e}_{i}^{t} \mathbf{e}_{i} \right)^{-1} \cdot \left( \mathbf{X}_{-i}^{t} \widehat{\boldsymbol{\alpha}} - \mathbf{X}_{i}^{t} - \lambda \mathbf{e}_{i}^{t} \right) \cdot \mathbf{e}_{i}^{t} \mathbf{Y} \nonumber \\
            &=& \mathbf{Y} - \mathbf{X}_{-i} \left( \mathbf{X}_{-i}^{t} \mathbf{X}_{-i} \right)^{-1} \mathbf{X}_{-i}^{t} \mathbf{Y} + \left( (\lambda + 1) \cdot \mathbf{e}_{i}^{t} \mathbf{e}_{i} \right)^{-1} \cdot (1 + \lambda) \mathbf{e}_{i} \cdot \mathbf{e}_{i}^{t} \mathbf{Y} \nonumber \\
            &=& \mathbf{Y} - \mathbf{X}_{-i} \left( \mathbf{X}_{-i}^{t} \mathbf{X}_{-i} \right)^{-1} \mathbf{X}_{-i}^{t} \mathbf{Y} - \mathbf{e}_{i} \cdot \frac{\mathbf{e}_{i}^{t} \mathbf{Y}}{\mathbf{e}_{i}^{t} \mathbf{e}_{i}}, \quad i=2,\dots,p. \label{residuos_alzado}
    \end{eqnarray}

    Taking into account the results included in Appendix \ref{appendix_inicial}, it is obtained that $\mathbf{e}(\lambda) = \mathbf{e}$. It is to say, the residuals of initial (\ref{modelo0}) and raised (\ref{modelo_alzado}) models coincide. It that case it is possible to conclude that:
    \begin{itemize}
        \item It is evident that the sum of squares of residual also coincide and, consequently, both models provide the same estimation for the variance of the random disturbance: $\widehat{\sigma}^{2} = \widehat{\sigma}^{2}(\lambda)$. 
        \item Since the explained variable is the same in both models, the sum of total squares coincide and, consequently, the coefficient of determination of both models is also the same: $R^{2} = R^{2}(\lambda)$.
        \item Since the statistic $F$ of the test of joint signification can be expressed as a function of the coefficient of determination, it is evident that the joint signification of both models coincide: $F_{exp} = F_{exp}(\lambda)$.
        \item These global characteristics (estimated variance of the random disturbance, the coefficient of determination and the statistic of the joint significance test) coincide also to the residualization (see Appendix \ref{appendix_ortogonal}).
        \item Since the observed values are the same, both models provide the same estimations for the dependent variable. Thus, when this technique is applied, the estimations are not modified while the degree of multicollinearity is mitigated as will be show in section \ref{multicol}.
    \end{itemize}

    \subsection{Individual inference}
        \label{inference_raise}

    In relation to the individual inference of model (\ref{modelo_alzado}), taking into account (see Appendix \ref{appendixA}) that:
    \begin{equation}
        \label{inferencia_alzada1}
        \left( \widetilde{\mathbf{X}}^{t} \widetilde{\mathbf{X}} \right)^{-1} = \left(
        \begin{array}{cc}
            \left( \mathbf{X}_{-i}^{t} \mathbf{X}_{-i} \right)^{-1} + \left( (\lambda + 1)^{2} \cdot \mathbf{e}_{i}^{t} \mathbf{e}_{i} \right)^{-1} \cdot \widehat{\boldsymbol{\alpha}} \widehat{\boldsymbol{\alpha}}^{t} & - \widehat{\boldsymbol{\alpha}} \cdot \left( (\lambda + 1)^{2} \cdot \mathbf{e}_{i}^{t} \mathbf{e}_{i} \right)^{-1} \\
            - \widehat{\boldsymbol{\alpha}}^{t} \cdot \left((\lambda + 1)^{2} \cdot \mathbf{e}_{i}^{t} \mathbf{e}_{i} \right)^{-1} & \left( (\lambda + 1)^{2} \cdot \mathbf{e}_{i}^{t} \mathbf{e}_{i} \right)^{-1}
        \end{array} \right),
    \end{equation}
    the null hypothesis  $H_{0}: \beta(\lambda)_{i} = 0$, with $i=2,\dots,p$, is rejected (in favor of the alternative hypothesis $H_{1}: \beta(\lambda)_{i} \not= 0$) when:
    \begin{equation}
        \label{inferencia_alzado1}
        t_{exp} (\beta(\lambda)_{i}) = \left| \frac{\widehat{\beta}(\lambda)_{i}}{\widehat{\sigma} \cdot \sqrt{\left( (\lambda + 1)^{2} \cdot \mathbf{e}_{i}^{t} \mathbf{e}_{i} \right)^{-1}}} \right| = \left| \frac{\mathbf{e}_{i}^{t} \mathbf{Y}}{\widehat{\sigma} \sqrt{\mathbf{e}_{i}^{t} \mathbf{e}_{i}}} \right| > t_{n-p}(1-\alpha/2),
    \end{equation}
    where $t_{n-p}(1-\alpha/2)$ is the value of a t-student distribution with $n-p$ degree of freedom that accumulate to its let a probability of $1-\alpha/2$ and $\widehat{\sigma}^{2}$ is the OLS estimation of the variance of the random disturbance of initial model (\ref{modelo0}) that, as previously commented, coincide to the one of the raised model (\ref{modelo_alzado}).

    Taking into account Appendix \ref{appendix_inicial}, it is obtained that $t_{exp} (\beta(\lambda)_{i}) = t_{exp}(\beta_{i})$. It is to say, the individual inference of the raised variable is unaltered.

    At the same time, the null hypothesis $H_{0}: \beta(\lambda)_{-i,j} = 0$ is rejected (in favor of the alternative hypothesis $H_{1}: \beta(\lambda)_{-i,j} \not= 0$) when:
    \begin{equation}
        t_{exp} (\beta(\lambda)_{-i,j}) = \left| \frac{\widehat{\beta}(\lambda)_{-i,j}}{\widehat{\sigma} \cdot \sqrt{w_{jj} + \left( (\lambda + 1)^{2} \cdot \mathbf{e}_{i}^{t} \mathbf{e}_{i} \right)^{-1} \cdot \widehat{\alpha}_{j}^{2}}} \right|
        = \left| \frac{\widehat{\gamma}_{-i,j} - (\lambda+1)^{-1} \cdot \widehat{\alpha}_{j} \widehat{\beta}_{i}}{\widehat{\sigma} \cdot \sqrt{w_{jj} + \left( (\lambda + 1)^{2} \cdot \mathbf{e}_{i}^{t} \mathbf{e}_{i} \right)^{-1} \cdot \widehat{\alpha}_{j}^{2}}} \right| > t_{n-p}(1-\alpha/2), \label{inferencia_alzado2}
    \end{equation}
    where $w_{jj}$ and $\widehat{\alpha}_{j}^{2}$ are the elements ($j$,$j$) of matrices $\left( \mathbf{X}_{-i}^{t} \mathbf{X}_{-i} \right)^{-1}$ and $\widehat{\boldsymbol{\alpha}} \widehat{\boldsymbol{\alpha}}^{t}$, respectively.
    Note that $\beta(\lambda)_{-i,j}$ is the element $j$ of $\boldsymbol{\beta}(\lambda)_{-i}$ with $j=1,\dots,i-1,i+1,\dots,p$.

    To summarize, the inference of the non raised variables changes as the raising factor varies. \cite{Salmeron2017} presents a simulation for $p=3$ where supposing a sample size equal to 20, it is concluded that in more than in the 50\% of the cases the null hypothesis is rejected. In addition, this percentage increases as the size of the sample increases reaching almost the 84\% for $n=160$. This fact is due to the effect that the sample size has on the variance of the estimated variables (see \cite{OBrien} for more details).

    Finally, taking into account the Appendix \ref{appendix_ortogonal}, it is obtained that:
    $$\lim \limits_{\lambda \rightarrow +\infty} t_{exp} (\beta(\lambda)_{-i,j}) = t_{exp} (\gamma_{-i,j}),$$
    it is to say, for $\lambda \rightarrow +\infty$ the inference of the non raised variables is the same in the raise and residualization.

\section{Multicollinearity}
    \label{multicol}

    Originally, the raise regression appears with the goal of estimating multiple linear regression models with worrying multicollinearity. Thus, its application should mitigate the multicollinearity.

    To check if this goal is reached, this section analyze as a novelty for any value of $p$, the norm of its estimations, the variance of the estimators and the usefulness of the raise regression to mitigate the essential and non-essential multicollinearity.

    This section also summarizes the results presented by \cite{Garcia2011}, \cite{Garcia2020}, \cite{Salmeron2020},  \cite{Roldan2019} and \cite{Roldan2020} that adapted the Variance Inflation Factor and the Condition Number, traditionally applied to measure the degree of multicollinearity, to be used after the application of the raise regression.

    \subsection{Estimator norm}
        \label{norm_raise}

    As known, one of the consequences of strong multicollinearity is the instability of the estimated coefficients. For this reason, it will be interesting to analyze the norm\footnote{If $\mathbf{x} \in \mathbb{R}$, then $\left\| \mathbf{x} \right\| = \sqrt{x_{1}^{2} + x_{2}^{2} + \dots + x_{n}^{2}}$. That is, $\left\| \mathbf{x} \right\|^{2} = \mathbf{x}^{t}\mathbf{x}$.} of the raise estimator (expression (\ref{estimador_raise})) in order to analyze if there is a value of the raising factor that stabilize the norm of the estimator analogously to that presented in \cite{HoerlKennard1970b} and \cite{HoerlKennard1970a} for the ridge estimation.

    Thus, from expression (\ref{estimador_raise}) it is obtained that:
    \begin{eqnarray}
        \left\| \widehat{\boldsymbol{\beta}}(\lambda) \right\|^{2} &=& \widehat{\boldsymbol{\beta}}(\lambda)^{t} \widehat{\boldsymbol{\beta}}(\lambda) = \widehat{\boldsymbol{\beta}}(\lambda)_{-i}^{t} \widehat{\boldsymbol{\beta}}(\lambda)_{-i} + \widehat{\beta}(\lambda)_{i}^{t} \widehat{\beta}(\lambda)_{i} \nonumber \\
            &=& \mathbf{Y}^{t} \mathbf{X}_{-i} \left( \mathbf{X}_{-i}^{t} \mathbf{X}_{-i} \right)^{-1} \cdot \left( \mathbf{X}_{-i}^{t} \mathbf{X}_{-i} \right)^{-1} \mathbf{X}_{-i}^{t} \mathbf{Y} - \mathbf{Y}^{t} \mathbf{X}_{-i} \left( \mathbf{X}_{-i}^{t} \mathbf{X}_{-i} \right)^{-1} \cdot \widehat{\boldsymbol{\alpha}} \cdot \frac{\mathbf{e}_{i}^{t} \mathbf{Y}}{(1+\lambda) \cdot \mathbf{e}_{i}^{t} \mathbf{e}_{i}} \nonumber \\
            & & - \frac{\mathbf{Y}^{t} \mathbf{e}_{i}}{(1+\lambda) \cdot \mathbf{e}_{i}^{t} \mathbf{e}_{i}}\cdot \widehat{\boldsymbol{\alpha}}^{t} \cdot \left( \mathbf{X}_{-i}^{t} \mathbf{X}_{-i} \right)^{-1} \mathbf{X}_{-i}^{t} \mathbf{Y} + \frac{\mathbf{Y}^{t} \mathbf{e}_{i}}{(1+\lambda) \cdot \mathbf{e}_{i}^{t} \mathbf{e}_{i}} \cdot \widehat{\boldsymbol{\alpha}}^{t} \widehat{\boldsymbol{\alpha}} \cdot \frac{\mathbf{e}_{i}^{t} \mathbf{Y}}{(1+\lambda) \cdot \mathbf{e}_{i}^{t} \mathbf{e}_{i}} \nonumber \\
            & & + \left( \frac{\mathbf{e}_{i}^{t} \mathbf{Y}}{(1+\lambda) \cdot \mathbf{e}_{i}^{t} \mathbf{e}_{i}} \right)^{2}. \label{norma_raise}
    \end{eqnarray}

    From expression (\ref{norma_raise}) it is obtained that:
    $$\lim \limits_{\lambda \rightarrow +\infty} \left\| \widehat{\boldsymbol{\beta}}(\lambda) \right\|^{2} = \mathbf{Y}^{t} \mathbf{X}_{-i} \left( \mathbf{X}_{-i}^{t} \mathbf{X}_{-i} \right)^{-1} \cdot \left( \mathbf{X}_{-i}^{t} \mathbf{X}_{-i} \right)^{-1} \mathbf{X}_{-i}^{t} \mathbf{Y} \underbrace{=}_{Appendix \ \ref{appendix_ortogonal}} \widehat{\boldsymbol{\gamma}}_{-i}^{t} \widehat{\boldsymbol{\gamma}}_{-i} = \left\| \widehat{\boldsymbol{\gamma}}_{-i} \right\|^{2}.$$
    It is to say, the norm of the raise estimator presents a horizonal asymptote  (that coincides to the one of the estimations of the non-residualizated variables in residualization). Thus, there must be a raising factor for which the norm of the estimator is stabilized.

    \subsection{Estimated variance}
        \label{var_beta_raise}

    Another symptom of multicollinearity in a multiple linear regression model is the inflation of the variance of the estimated parameters. This subsection analyzes the effect of the raised regression on these variances.

    Thus, by comparing the elements of the main diagonal of expressions (\ref{inferencia_alzada1}) and (\ref{eq_anexo_2}) it is obtained that both coincide when  $\lambda = 0$ and the first diminish as the value of $\lambda$ increases. Then, the raised regression diminishes the variances of the estimated coefficients which are supposed to be inflated as a consequence of the multicollinearity. This fact shows the usefulness of this technique to mitigate the multicollinearity.

     In addition, taking into account the expression of $\left(\mathbf{X}_{O}^{t} \mathbf{X}_{O} \right)^{-1}$ included in expression (\ref{est_ortogonal}), it is observed that for $\lambda \rightarrow + \infty$ the elements of the main diagonal of (\ref{inferencia_alzada1}) for the non-raised variables coincide to the one of $\left(\mathbf{X}_{O}^{t} \mathbf{X}_{O} \right)^{-1}$. Once again, it is showed the asymptotic relation between the raise and the residualization technique.

    \subsection{Variance Inflation Factor}
        \label{vif_raise}

    Given the model (\ref{modelo0}), the  Variance Inflation Factor (VIF) is obtained as:
    \begin{equation}
        VIF(i) = \frac{1}{1 - R_{i}^{2}}, \quad i=2,\dots,p,
        \label{vif}
    \end{equation}
    where $R_{i}^{2}$ is the coefficient of determination of auxiliary regression (\ref{modelo_aux}).
    If the variable $\mathbf{X}_{i}$ has not linear relation (it is orthogonal) with the rest of independent variables, it is obtained that $R_{i}^{2} = 0$ and then $VIF(i) = 1$. Thus, the $R_{i}^{2}$ (and consequently the $VIF(i)$) increases as the relation between the variables is stronger. Then, the higher the value of the VIF associated to the variable $\mathbf{X}_{i}$ the higher the linear relation between this variable and the rest of the independent variables of the model (\ref{modelo0}). For values of the VIF higher than 10, it is considered that the presence of multicollinearity is worrying (see, for example, \cite{Marquardt1970} or \cite{Salmeron2018}).

    Analogously, the VIF in the raise regression is given by:
    \begin{equation}
        VIF(l,\lambda) = \frac{1}{1 - R_{l}^{2}(\lambda)}, \quad l=2,\dots,p,
        \label{vif_raised}
    \end{equation}
    where to calculate $R_{l}^{2}(\lambda)$ it is necessary to distinguish two cases depending of which will be the dependent variable of the auxiliary regression:  a) if it is the raised variable, $\widetilde{\mathbf{X}}_{i}$ with $i=2,\dots,p$, or b) if it is not the raised variable, $\mathbf{X}_{j}$ with $j=2,\dots,p$ being $j \not= i$.

    \cite{Garcia2011} shows that the VIF associated to the raise variable diminishes as the raising factor increases, while \cite{Garcia2020} shows that the VIFs of the non-raised variables also diminish as the raising factor increases.

    Finally, these results are completed with the contribution of \cite{Salmeron2020}, where it is shown that the $VIF(l,\lambda)$ has the very desirable properties of a) being continuous ($VIF(l,0) = VIF(l)$ with $l=2,\dots,p$), b) monotone in the raise parameter ($VIF(l,\lambda)$ decrease if $\lambda$ increase) and c) higher than one ($VIF(l,\lambda) \geq 1$, for $\lambda \geq 0$ and $l=2,\dots,p$).
    In addition, the following conditions are verified:
    \begin{itemize}
        \item When $\lambda \rightarrow +\infty$, the VIF associated to the raised variable tends to one.
        \item When $\lambda \rightarrow +\infty$, the VIFs associated to the not raised variables present a horizonal asymptote given by the VIF of the regression $\mathbf{Y} = \mathbf{X}_{-i} \boldsymbol{\xi} + \mathbf{w}$. Note that in this case, the limit of the VIF of the non raised variables coincide to the unaltered variables of the residualization (see Appendix \ref{appendix_ortogonal}).
    \end{itemize}

    \subsection{Condition Number}
        \label{cn_raise}

    Given the model (\ref{modelo0}), the Condition Number (CN) is defined as:
	\begin{equation}
	   K(\mathbf{X})=\sqrt{\frac{\xi_{max}}{\xi_{min}}}, \label{cn}
	\end{equation}
	where $\xi_{max}$ and $\xi_{min}$ are the maximum and the minimum eigenvalues of the matrix $\mathbf{X}^{t} \mathbf{X}$, respectively. Note that before calculating the eigenvalues of the matrix $\mathbf{X}$, it is necessary to transform it to present unit length columns. Values of CN between 20 and 30 indicate moderate multicollinearity while values higher than 30 indicate strong multicollinearity (see, for example, \cite{Belsley1980} or \cite{Salmeron2018}).

    Analogously, the CN of the raise regression is given by the following expression:
    \begin{equation}
        K(\widetilde{\mathbf{X}},\lambda)=\sqrt{\frac{\widetilde{\xi}_{max}}{\widetilde{\xi}_{min}}},
        \label{cn_raised}
    \end{equation}
    where $\widetilde{\xi}_{max}$ and $\widetilde{\xi}_{min}$ are, respectively, the maximum and minimum eigenvalues of matrix $\widetilde{\mathbf{X}}^{t} \widetilde{\mathbf{X}}$.

     \cite{Roldan2019} analyzed with detail the condition number $K(\widetilde{\mathbf{X}},\lambda)$ for $p=3$ in the model (\ref{modelo0}) and standardizing the independent variables (subtracting its mean and dividing by the square root of $n$ times its variance). \cite{Roldan2020} generalized these results for any value of $p$ and without transforming the variables obtaining that $K(\widetilde{\mathbf{X}},\lambda)$ has the very desirable properties of a) being continuous ($K(\widetilde{\mathbf{X}},0) = K(\mathbf{X})$), b) monotone in the raise parameter ($K(\widetilde{\mathbf{X}},\lambda)$ decrease if $\lambda$ increase) and c) higher than one ($K(\widetilde{\mathbf{X}},\lambda) \geq 1$ for $\lambda \geq 0$) since when $\lambda \rightarrow +\infty$, $K(\widetilde{\mathbf{X}},\lambda)$ presents a horizontal asymptote that coincide to the CN of the matrix  $\mathbf{X}_{-i}$. Note that in this case, the limit of the CN coincide to the CN of the residualization (see Appendix \ref{appendix_ortogonal}).

    \subsection{Essential and non-essential multicollinearity}
        \label{coef_variation}

    \cite{MarquardtSnee}, \cite{Marquardt1980} and \cite{SneeMarquardt} distinguished between essential and non-essential multicollinearity. The first one is the existence of a linear relation between the independent variables of model (\ref{modelo0}) excluding the intercept, while the second one is the linear relation between the intercept and at least one of the independent variables of the model.
    In \cite{Salmeron2020centered} is presented a definition of nonessential multicollinearity ``that generalizes the definition given by Marquardt and Snee. Note that this generalization can be understood as a particular kind of essential multicollinearity: a near-linear relation between two independent variables with light variability. However, it is shown that this kind of multicollinearity is not detected by the VIF, and for this reason, we consider it more appropriate to include it within the nonessential multicollinearity''.

    \cite{Salmeron2018} and \cite{Salmeron2020a} showed that the VIF is only able to detect the essential multicollinearity. Therefore, from results of subsection \ref{vif_raise} it is concluded that the raise regression mitigates this kind of multicollinearity. It is also possible to conclude that the raise regression mitigates the non-essential multicollinearity due to the CN also diminishes when the raise regression is applied (see subsection \ref{cn_raise}) being this measure able to detect both kind of multicollinearity. However, this question, that raise regression mitigates non-essential multicollinearity, is analyzed with more detail below.

    \cite{Salmeron2020a} showed that the most adequate measure to detect non-essential multicollinearity is the coefficient of variation (CV) of each one of the independent variables providing the following threshold: values of the coefficient of variation lower than 0.1002506 will indicate that the degree of non-essential multicollinearity is worrying. Thus, if it is shown that the coefficient of variance increases when the raise regression is applied, it will show that the raise regression mitigates the non-essential multicollinearity.

    Indeed, taking into account that the $CV (\widetilde{\mathbf{X}}_{i}) = \frac{\sqrt{var(\widetilde{\mathbf{X}}_{i})}}{\overline{\widetilde{\mathbf{X}}}_{i}}$ for $i=2,\dots,p$ and that:
    \begin{itemize}
        \item $var(\widetilde{\mathbf{X}}_{i}) = var(\mathbf{X}_{i} + \lambda \cdot \mathbf{e}_{i}) = var(\mathbf{X}_{i}) + \lambda^{2} \cdot var(\mathbf{e}_{i}) + 2 \cdot \lambda \cdot cov(\mathbf{X}_{i}, \mathbf{e}_{i}) = var(\mathbf{X}_{i}) + (\lambda^{2} + 2 \cdot \lambda) \cdot var(\mathbf{e}_{i})$ due to $cov(\mathbf{X}_{i}, \mathbf{e}_{i}) = \frac{1}{n} \cdot \mathbf{X}_{i}^{t} \mathbf{e}_{i} = \frac{1}{n} \cdot (\mathbf{e}_{i}^{t} \mathbf{e}_{i} + \widehat{\mathbf{X}}_{i}^{t} \mathbf{e}_{i}) = \frac{1}{n} \cdot \mathbf{e}_{i}^{t} \mathbf{e}_{i}$ where $\widehat{\mathbf{X}}_{i}$ is the estimation of $\mathbf{X}_{i}$ obtained from the auxiliary regression (\ref{modelo_aux}) and, consequently, it is verified that $\widehat{\mathbf{X}}_{i}^{t} \mathbf{e}_{i} = 0$, and
        \item $\overline{\widetilde{\mathbf{X}}}_{i} = \overline{\mathbf{X}}_{i} + \lambda \cdot \overline{\mathbf{e}}_{i} = \overline{\mathbf{X}}_{i}$,
    \end{itemize}
    it is verified that:
    $$CV (\widetilde{\mathbf{X}}_{i}) = \frac{\sqrt{var(\mathbf{X}_{i}) + (\lambda^{2} + 2 \cdot \lambda) \cdot var(\mathbf{e}_{i})}}{\overline{\mathbf{X}}_{i}}, \quad i \geq 2.$$

   In this case, since $(\lambda^{2} + 2 \cdot \lambda) \cdot var(\mathbf{e}_{i}) \geq 0$, then  $CV (\widetilde{\mathbf{X}}_{i}) \geq CV(\mathbf{X}_{i}) = \frac{\sqrt{var(\mathbf{X}_{i})}}{\overline{\mathbf{X}}_{i}}$ for $\lambda \geq 0$.

\section{Choice of variable to raise and the raise parameter}
    \label{ChoiceVarLanda}

    From the results of previous sections, the possible criteria to select the variable to be raised are enumerated:
    \begin{itemize}
        \item To raise the variable whose parameter will be initially different from zero, taking into account that this characteristic will be maintained (see subsection \ref{inference_raise}).
        \item If none estimated parameters are not individually significant, it is recommendable to raise the variable considered less important since the not raised coefficient variable (considered more relevant) can become individually significant after the application of the raise estimation (see subsection \ref{inference_raise}).
        \item To raise the variable with the higher VIF (see subsection \ref{vif_raise}) since it limit tends to 1.
        \item To raise the variable that after being raised presents a lower value for the asymptote of the VIF or the CN (see subsections \ref{vif_raise} and \ref{cn_raise}).
        \item To raise the variable with the lower CV (see subsection \ref{coef_variation}) since it will increase.
    \end{itemize}

    In relation to the criteria to select the raising factor we propose the following possibilities:
    \begin{itemize}
        \item To select the value that minimizes the MSE or the one that provides a value of MSE lower than the one provided by OLS (see subsection \ref{mse_raise}).
        \item To select the value that stabilize the norm of the raise estimator (see subsection \ref{norm_raise}).
        \item To select, if it exists, the value that provides values for the VIF or the CN lower or for CV higher than the established thresholds to determine the existence of worrying multicollinearity.
    \end{itemize}

   To finish, note that the criteria to select the variable to be raised or the raising factor can be not only one of the previous proposals but a combination of them.

\section{Successive raise}
    \label{SucRaise}
     \label{MtxMSE}

    Section \ref{multicol} showed that may not exist the raising factor that get the enough mitigation of the multicollinearity since the VIF and the CN present a lower bound. For this reason, it can be needed to raised more than one variable.

    \cite{Garcia2017}  described the successive raise that we summarize here with the following steps:
    \begin{description}
      \item[Step 1:] Raise the vector $\mathbf{X}_2$ from the regression of $\mathbf{X}_2$ with the resting vectors $\{ \mathbf{1}, \mathbf{X}_3, \mathbf{X}_4, \ldots, \mathbf{X}_p \}$. From this auxiliary regression, it is we constructed the raise vector of $\mathbf{X}_2$, $\widetilde{\mathbf{X}}_2$, as $\mathbf{X}_2+\lambda_2 \cdot \mathbf{e}_2$ where $\mathbf{e}_2$ is the residual vector of the auxiliary regression. Now, it is described a new space of vectors where $\mathbf{X}_2$ is replaced by $\widetilde{\mathbf{X}}_2$, $\{ \mathbf{1}, \widetilde{\mathbf{X}}_2, \mathbf{X}_3, \mathbf{X}_4, \ldots, \mathbf{X}_p \}$.
      \item[Step 2:] Raise the vector $\mathbf{X}_3$ from the regression of $\mathbf{X}_3$ with the resting vectors $\{ \mathbf{1}, \widetilde{\mathbf{X}}_2, \mathbf{X}_4, \ldots, \mathbf{X}_p \}$. From this auxiliary regression, it is we constructed the raise vector of $\mathbf{X}_3$, $\widetilde{\mathbf{X}}_3$, as $\mathbf{X}_3+\lambda_3 \cdot \mathbf{e}_3$ where $\mathbf{e}_3$ is the residual vector of the auxiliary regression. Now, it is described a new space of vectors where $\mathbf{X}_3$ is replaced by $\widetilde{\mathbf{X}}_3$, $\{ \mathbf{1}, \widetilde{\mathbf{X}}_2, \widetilde{\mathbf{X}}_3, \mathbf{X}_4, \ldots, \mathbf{X}_p \}$.
      \item[Step $j$:] Raise the vector $\mathbf{X}_j$ from the regression of $\mathbf{X}_j$ with the resting vectors $\{ \mathbf{1}, \widetilde{\mathbf{X}}_2, \widetilde{\mathbf{X}}_3, \ldots, \widetilde{\mathbf{X}}_{j-1}, \mathbf{X}_{j+1}, \ldots, \mathbf{X}_p \}$. From this auxiliary regression, it is we constructed the raise vector of $\mathbf{X}_j$, $\widetilde{\mathbf{X}}_j$, as $\mathbf{X}_j+\lambda_j \cdot \mathbf{e}_j$ where $\mathbf{e}_j$ is the residual vector of the auxiliary regression. Now, it is described a new space of vectors where $\mathbf{X}_j$ is replaced by $\widetilde{\mathbf{X}}_j$, $\{ \mathbf{1}, \widetilde{\mathbf{X}}_2, \widetilde{\mathbf{X}}_3, \widetilde{\mathbf{X}}_4, \ldots, \widetilde{\mathbf{X}}_j, \ldots, \mathbf{X}_p \}$.
      \item[Step $p$:] Raise the vector $\mathbf{X}_p$ from the residual vector, $\mathbf{e}_{p}$, of the regression of $\mathbf{X}_p$ with the set of vectors $\{ \mathbf{1}, \widetilde{\mathbf{X}}_2, \widetilde{\mathbf{X}}_3, \widetilde{\mathbf{X}}_4, \ldots, \widetilde{\mathbf{X}}_{p-1} \}$ as $\mathbf{X}_p+\lambda_p \cdot \mathbf{e}_p$. Finally, it is described a new space of vectors $\{ \mathbf{1}, \widetilde{\mathbf{X}}_2, \widetilde{\mathbf{X}}_3, \widetilde{\mathbf{X}}_4, \ldots, \widetilde{\mathbf{X}}_p \}$.
    \end{description}

    Note that when all the variables are raised, the VIFs and the CN tend to 1 and then it is assured that there is a value for the vector of raised parameter, $(\lambda_{2}, \dots, \lambda_{p})$, that makes the value of the VIFs and CN lower than the established thresholds assuming that the multicollinearity has been mitigated.

\section{Relation between raise and ridge regression}
    \label{MtxMSE}

    In \cite{Garcia2020} it is justified the ridge regression from a geometrical perspective and it is also formally related with the raise regression.  More concretely, that paper shows that ridge estimator, $\widehat{\boldsymbol{\beta }}_{R}(k)$, is a particular case of raising procedures shown in that paper justifying the presence of the well-known constant $k$  on the main diagonal of matrix $\mathbf{X}^{t} \mathbf{X}$.  This is to say, there is a combination of the vector of raised parameters $(\lambda_{2}, \dots, \lambda_{p})$, that verifies $\widetilde{\mathbf{X}}^{t}\widetilde{\mathbf{X}}=\mathbf{X}^{t}\mathbf{X}+k \cdot \mathbf{I}$. 

In this context, see Appendix \ref{appendix_mse}, from Proposition \ref{Proposition 1} it is possible to compare
estimators $\widehat{\boldsymbol{\beta }}$, $\widehat{\boldsymbol{\beta }}%
(\lambda)$, $\widehat{\boldsymbol{\beta }}_{R}(k)$ under the criteria of the
root mean square error matrix and MSE.

\begin{proposition}
\label{Proposition 2} The successive raise estimator, $\widehat{\boldsymbol{%
\beta }}(\lambda)=\left( \widetilde{\mathbf{X}}^{t}\widetilde{%
\mathbf{X}}\right) ^{-1}\widetilde{\mathbf{X}}^{t}\mathbf{Y%
}$ where $\widetilde{\mathbf{X}}^{t}\widetilde{\mathbf{X}}=%
\mathbf{X}^{t}\mathbf{X}+k\mathbf{I}$, is preferred to the OLS
estimator, $\widehat{\boldsymbol{\beta }}=\left( \mathbf{X}^{t}\mathbf{%
X}\right) ^{-1}\mathbf{X}^{t}\mathbf{Y}$, under the root mean square error matrix
criterion for values of\textrm{\ }$k$ that verify the following expression:
\begin{equation*}
\boldsymbol{\beta }^{t}\left[ \left( \widetilde{\mathbf{X}}%
^{t}\widetilde{\mathbf{X}}\right) ^{-1}\widetilde{\mathbf{X%
}}^{t}\mathbf{X}-\mathbf{I}\right] ^{t}\frac{1}{k}\left(
\mathbf{X}^{t}\mathbf{X}+k\mathbf{I}\right) \mathbf{X}^{t}%
\mathbf{X}\left[ \left( \widetilde{\mathbf{X}}^{t}\widetilde{%
\mathbf{X}}\right) ^{-1}\widetilde{\mathbf{X}}^{t}\mathbf{X%
}-\mathbf{I}\right] \boldsymbol{\beta }<\sigma ^{2}.
\end{equation*}
\end{proposition}

\begin{proof}
By considering $\mathbf{C}_{1}=\left( \widetilde{\mathbf{X}}^{t}%
\widetilde{\mathbf{X}}\right) ^{-1}\widetilde{\mathbf{X}}%
^{t}=\left( \mathbf{X}^{t}\mathbf{X}+k\mathbf{I}\right)
^{-1}\widetilde{\mathbf{X}}^{t}$ and $\mathbf{C}_{2}=\left(
\mathbf{X}^{t}\mathbf{X}\right) ^{-1}\mathbf{X}^{t}$, and only
replacing $\mathbf{C}_{1}$ and $\mathbf{C}_{2}$, it is possible to show $%
\mathbf{C}_{2}\mathbf{C}_{2}^{t}-\mathbf{C}_{1}\mathbf{C}_{1}^{\prime
}=\left( \mathbf{X}^{t}\mathbf{X}\right) ^{-1}-\left( \mathbf{X}%
^{t}\mathbf{X}+k\mathbf{I}\right) ^{-1}$. It is also verified that
this matrix is positive definite (see \cite{Puntanen2011}). On the other
hand, remembering that the inverse of the sum of two matrices $\mathbf{A}$
and $\mathbf{B}$ is defined as:
\begin{equation}
\left( \mathbf{A}+\mathbf{B}\right) ^{-1}=\mathbf{A}^{-1}-\left( \mathbf{I}+%
\mathbf{A}^{-1}\mathbf{B}\right) ^{-1}\mathbf{A}^{-1}\mathbf{B}\mathbf{A}%
^{-1},  \label{eq: inversa de la suma de matrices}
\end{equation}%
it is possible to show that:
\begin{equation*}
\left( \mathbf{C}_{2}\mathbf{C}_{2}^{t}-\mathbf{C}_{1}\mathbf{C}%
_{1}^{t}\right) ^{-1}=\left[ \left( \mathbf{X}^{t}\mathbf{X}%
\right) ^{-1}-\left( \mathbf{X}^{t}\mathbf{X}+k\mathbf{I}\right) ^{-1}%
\right] ^{-1}=k\left( \mathbf{X}^{t}\mathbf{X}+k\mathbf{I}\right)
^{-1}\left( \mathbf{X}^{t}\mathbf{X}\right) ^{-1}.
\end{equation*}%
Given that $\mathbf{X}^{t}\mathbf{X}+k\mathbf{I}$ and $\mathbf{X}^{t}\mathbf{X}$ are positive definite
commutable matrices and $k$ is positive, then their product will be a
positive definite matrix. Based on the result shown in Proposition \ref%
{Proposition 1}, it is therefore possible to conclude that $\widehat{\boldsymbol{\beta }}(\lambda)$ satisfies the MSE admissibility condition, thus
ensuring an improvement in MSE for $k>0$.
\end{proof}
\begin{proposition}
The ridge estimator $\widehat{\boldsymbol{\beta }}_{R}(k)=\left( \mathbf{X}%
^{t}\mathbf{X}+k\mathbf{I}\right) ^{-1}\mathbf{X}^{t}\mathbf{Y}$
is preferred to the successive raise estimator, $\widehat{\boldsymbol{\beta }%
}(\lambda)=\left( \widetilde{\mathbf{X}}^{t}\widetilde{\mathbf{X}}%
\right) ^{-1}\widetilde{\mathbf{X}}^{t}\mathbf{Y}$ where $%
\widetilde{\mathbf{X}}^{t}\widetilde{\mathbf{X}}=\mathbf{X}%
^{t}\mathbf{X}+k\mathbf{I}$, under the criteria of the root mean
square error matrix for the values of $k$ that verify the following
expression:
\begin{equation*}
\boldsymbol{\beta }^{t}\left[ \left( \mathbf{X}^{t}\mathbf{X}+k%
\mathbf{I}\right) ^{-1}\mathbf{X}^{t}\mathbf{X}-\mathbf{I}\right]
^{t}\frac{1}{k}\left( \mathbf{X}^{t}\mathbf{X}+k\mathbf{I}%
\right) \left[ \left( \mathbf{X}^{t}\mathbf{X}+k\mathbf{I}\right) ^{-1}%
\mathbf{X}^{t}\mathbf{X}-\mathbf{I}\right] \boldsymbol{\beta }<\sigma
^{2}.
\end{equation*}
\end{proposition}

\begin{proof}
By defining $\mathbf{C}_{1}=\left( \mathbf{X}^{t}\mathbf{X}+k\mathbf{I}%
\right) ^{-1}\mathbf{X}^{t}$ and $\mathbf{C}_{2}=\left( \widetilde{%
\mathbf{X}}^{t}\widetilde{\mathbf{X}}\right) ^{-1}%
\widetilde{\mathbf{X}}^{t}=\left( \mathbf{X}^{t}\mathbf{X}%
+k\mathbf{I}\right) ^{-1}\widetilde{\mathbf{X}}^{t}$, it is
possible to show that $$\mathbf{C}_{2}\mathbf{C}_{2}^{t}-\mathbf{C}_{1}%
\mathbf{C}_{1}^{t}=\left( \mathbf{X}^{t}\mathbf{X}+k\mathbf{I}%
\right) ^{-1}\left[ \mathbf{I}-\mathbf{X}^{t}\mathbf{X}\left( \mathbf{X%
}^{t}\mathbf{X}+k\mathbf{I}\right) ^{-1}\right].$$ From
\eqref{eq: inversa de la
suma de matrices}, it is obtained that:
\begin{eqnarray}
\left( \mathbf{X}^{t}\mathbf{X}+k\mathbf{I}\right) ^{-1} &=&\left(
\mathbf{X}^{t}\mathbf{X}\right) ^{-1}-\left[ \mathbf{I}+\left( \mathbf{%
X}^{t}\mathbf{X}\right) ^{-1}k\right] ^{-1}\left( \mathbf{X}^{t}%
\mathbf{X}\right) ^{-1}k\left( \mathbf{X}^{t}\mathbf{X}\right) ^{-1}
\notag \\
&=&\left( \mathbf{X}^{t}\mathbf{X}\right) ^{-1}-\left[ k\left( \mathbf{%
X}^{t}\mathbf{X}\right) ^{-1}+\mathbf{I}\right] ^{-1}\left( \mathbf{X}%
^{t}\mathbf{X}\right) ^{-1}k\left( \mathbf{X}^{t}\mathbf{X}%
\right) ^{-1}.  \notag
\end{eqnarray}%
Then:
\begin{eqnarray}
\left[ \mathbf{I}-\mathbf{X}^{t}\mathbf{X}\left( \mathbf{X}^{t}%
\mathbf{X}+k\mathbf{I}\right) ^{-1}\right]  &=&k\left( \mathbf{X}^{t}%
\mathbf{X}\right) ^{-1}\left[ k\left( \mathbf{X}^{t}\mathbf{X}\right)
^{-1}+\mathbf{I}\right] ^{-1}\left( \mathbf{X}^{t}\mathbf{X}\right)
^{-1}\left( \mathbf{X}^{t}\mathbf{X}\right) ^{-1}  \notag \\
&=&k\left( \mathbf{X}^{t}\mathbf{X}\right) \left( \left( \mathbf{X}%
^{t}\mathbf{X}\right) \left[ k\left( \mathbf{X}^{t}\mathbf{X}%
\right) ^{-1}+\mathbf{I}\right] \right) ^{-1}\left( \mathbf{X}^{t}%
\mathbf{X}\right) ^{-1}  \notag \\
&=&k\left( \mathbf{X}^{t}\mathbf{X}\right) \left[ k\mathbf{I}+\left(
\mathbf{X}^{t}\mathbf{X}\right) \right] ^{-1}\left( \mathbf{X}^{\prime
}\mathbf{X}\right) ^{-1}  
= k\left( \mathbf{X}^{t}\mathbf{X}\right) \left( \left( \mathbf{X}%
^{t}\mathbf{X}\right) \left[ k\mathbf{I}+\left( \mathbf{X}^{t}%
\mathbf{X}\right) \right] \right) ^{-1}  \notag \\
&=&k\left( \mathbf{X}^{t}\mathbf{X}\right) \left[ \left( \mathbf{X}%
^{t}\mathbf{X}\right) ^{2}+k\left( \mathbf{X}^{t}\mathbf{X}%
\right) \right] ^{-1}  
= k\left( \left[ \left( \mathbf{X}^{t}\mathbf{X}\right) ^{2}+k\left(
\mathbf{X}^{t}\mathbf{X}\right) \right] \left( \mathbf{X}^{t}%
\mathbf{X}\right) ^{-1}\right) ^{-1}  \notag \\
&=&k \left( \mathbf{X}^{t} \mathbf{X} +k\mathbf{I}\right)^{-1}.  \notag
\end{eqnarray}%
In conclusion, it is obtained that $\mathbf{C}_{2}\mathbf{C}_{2}^{t}-%
\mathbf{C}_{1}\mathbf{C}_{1}^{t}=k \left( \mathbf{X}^{t}%
\mathbf{X} +k\mathbf{I}\right)^{-1}$ will be a positive definite
matrix if $k>0$. Based on Proposition \ref{Proposition 1}, we conclude that,
in relation to the $MtxMSE$\ and MSE criteria, the ridge estimator is
preferred to estimator $\widehat{\boldsymbol{\beta }}(\lambda)$\ for $k\in
(0,+\infty )$.
\end{proof}

\section{Empirical examples}
    \label{example}

To illustrate the empirical application of the raise regression, we consider two data sets.  The first one is used to show the estimation of the raise regression and compare it with the ridge and residualization regressions. It is also shown that the raise and residualization regressions maintains all the global characteristics of the OLS regression while, at the same time, mitigate the multicollinearity. The second example is used to illustrate the usefulness of the raise regression to mitigate the non-essential multicollinearity.

\subsection{Example 1: Portland cement dataset}

The regression model for these data is defined as:
\begin{equation}
    \mathbf{Y} =\beta_1+\beta_2 \mathbf{X}_{2}+\beta_3 \mathbf{X}_{3}+ \beta_4 \mathbf{X}_{4}+\beta_5 \mathbf{X}_{5}+ \mathbf{u},
\end{equation}
being $\mathbf{Y}$ the heat evolved after 180 days of curing measured in calories per gram of cement, $\mathbf{X}_{2}$ the tricalcium aluminate, $\mathbf{X}_{3}$ the tricalcium silicate, $\mathbf{X}_{4}$ the tetracalcium aluminoferrite and $\mathbf{X}_{5}$ the $\beta$-dicalcium silicate.

This example is originally adopted by \cite{woods1932effect} and also analyzed by, among others, \cite{kacciranlar1999new}, \cite{li2012new} and, recently, by \cite{lukman2019modified} and \cite{kibria2020new}.

Table \ref{correxample1} presents the main measures in relation to the diagnose of multicollinearity: the correlation matrix of the variables, the coefficients of variation, the variance inflation factors obtained from OLS and the condition number (obtained from the matrix $\mathbf{X}^{t} \mathbf{X}$ transformed to present unit length columns\footnote{Note that \cite{kibria2020new} obtained a condition number approximately equal to 424 directly dividing the maximum eigenvalue of the original data (44676.206) by the minimum eigenvalue of the original data (105.419) and with out calculate the square root. \cite{kacciranlar1999new} obtained a condition number equal to 20.58 by using the original data and calculating the square root.}). All these measures indicate the presence of a troubling degree of essential multicollinearity in this model.

\begin{table}
\begin{center}
\begin{tabular}{cccccc}
  \hline
  &         Y&$\mathbf{X}_{2}$&        $\mathbf{X}_{3}$  &      $\mathbf{X}_{4}$&         $\mathbf{X}_{5}$\\
  Y             &1.0000000  &0.7307175  &0.8162526  &-0.5346707  &-0.8213050\\
$\mathbf{X}_{2}$&0.7307175  & 1.0000000 & 0.2285795 &-0.8241338  &-0.2454451\\
$\mathbf{X}_{3}$&0.8162526  & 0.2285795 &1.0000000  &-0.1392424  &-0.9729550\\
$\mathbf{X}_{4}$&-0.5346707 &-0.8241338 &-0.1392424 & 1.0000000  & 0.0295370\\
$\mathbf{X}_{5}$&-0.8213050 &-0.2454451 &-0.9729550 &  0.0295370 &1.0000000\\
\hline
VIF &&38.49621&254.42317&46.868399&282.5128\\
\hline
CV&&0.7574338&0.3104718&0.5228758&0.5360508\\
CN&249.5783&&&&\\
\hline
\end{tabular}
\caption{Main measures in relation to the diagnose of multicollinearity for Portland cement dataset}\label{correxample1}
\end{center}
\end{table}

Results of the estimation by OLS are provided in the first column of Table \ref{resultsexample1}. Note that only the regression coefficient associated to variable $\mathbf{X}_{2}$ is significant at the level of confidence of 90\%. In contrast, the coefficient of determination is 98.24\% and the model is globally significant at the level of confidence of 99\%. In fact, \cite{kacciranlar1999new} considers that this is an example of an ``anti-quirk'' due to it is verified that $R^2< \sum_{i=2}^{5} r^2_{x_i,y}$ being $r^2_{x_i,y}$ the simple correlation between $x_{i}$ and $y$ with $i=2,3,4,5$. This fact is also considered as a symptom of multicollinearity.

\begin{table}
\begin{center}
\begin{tabular}{cc|cccc}
  \hline
   & OLS& &RIDGE ($k=0.0077$)& RIDGE ($k=0.0015$) & RIDGE ($k=20.6$)\\
   \hline
  $\hat{\beta}_1$& 62.4054 (70.071) &$\hat{\beta}_1(k)$& 8.5642 & 27.9917& \\

  $\hat{\beta}_2$  & 1.5511 (0.7448)*&$\hat{\beta}_2(k)$ & 2.1048 &1.9051&1.6757 \\

  $\hat{\beta}_3$ & 0.5102 (0.7238) &$\hat{\beta}_3(k)$& 1.0651 & 0.8648&1.8592\\

  $\hat{\beta}_4$ & 0.1019 (0.7547)&$\hat{\beta}_4(k)$& 0.6683 & 0.4640&-1.211\\

  $\hat{\beta}_5$  & -0.1441 (0.7091)&$\hat{\beta}_5(k)$&0.3998 &0.2036&-1.8771\\

  \hline
  $R^2$& 0.9824&&&\\
  $F_{statistic}$&111.5***&&&\\
  MSE& 4912.09&&2991.8&2171.295&1711.66\\
  \hline
\end{tabular}
\caption{OLS and ridge estimations for Portland cement dataset}\label{resultsexample1}
\end{center}
\end{table}

Due to the existence of multicollinearity, Table \ref{resultsexample1} presents also the estimation by ridge regression with the value of $k$ proposed by \cite{HoerlKennardBaldwin} ($k=0.0077$) and the value of $k$ that minimizes the MSE ($k=0.0015$). The last column presents the estimation by ridge regression with standardized data with the value of $k$ that minimizes the MSE ($k=20.6$). 

\begin{table}
\begin{center}
\begin{tabular}{cccccc}
  \hline
Raised& Variable $\mathbf{X}_{2}$& Variable $\mathbf{X}_{3}$& Variable $\mathbf{X}_{4}$& Variable $\mathbf{X}_{5}$\\
 \hline
$\lim_{\lambda\rightarrow\infty}VIF(X_2,\lambda)$& 1        &3.6781 &1.06633    &3.2510\\
$\lim_{\lambda\rightarrow\infty}VIF(X_3,\lambda)$&24.3091   &1      &18.78031    &1.06357\\
$\lim_{\lambda\rightarrow\infty}VIF(X_4,\lambda)$&1.298236   &3.4596 &1          &3.1421\\
$\lim_{\lambda\rightarrow\infty}VIF(X_5,\lambda)$&23.8586   &1.1810 &18.9400   &1\\
 \hline
 $\lim_{\lambda\rightarrow\infty} K( \widetilde{\mathbf{X}},\lambda)$  & 61.01328   &  15.30453  &  51.51909 & 13.65739  \\
\hline
$\lambda_{min}$  &0.230547&2.01277&54.8438&24.2248\\
$MSE \left( \widehat{\boldsymbol{\beta}}(\lambda_{min}) \right)$&4050.054&1645.21&308.5229&211.2751\\
\hline
\end{tabular}
\caption{Horizontal asymptotes for VIF an d CN after raising each variable and $\lambda_{min}$ (and its MSE) for Portland cement dataset}\label{landaexample1}
\end{center}
\end{table}

Table \ref{landaexample1} shows the horizontal asymptote for VIFs and CN after raising each variable and also the value of $\lambda_{min}$, the $\lim_{\lambda\rightarrow\infty} K( \widetilde{\mathbf{X}},\lambda)$ and the $MSE \left( \widehat{\boldsymbol{\beta}}(\lambda_{min}) \right)$. 
In relation to the selection of the variable to be raised, note that:
\begin{itemize}
    \item When variable $\mathbf{X}_{2}$ or $\mathbf{X}_{4}$ are raised, there is an asymptote higher than 10 for the VIFs of variables $\mathbf{X}_{3}$ and $\mathbf{X}_{5}$. However, the thresholds that would be obtained for VIFs by raising variables $\mathbf{X}_{3}$ or $\mathbf{X}_{5}$ are in all cases less than 10.
    \item In relation to $\lambda_{min}$ in all cases the value is higher than one except for the case when the second variable is raised. Note that when $\lambda_{min} > 1$  it is verified that $MSE(\widehat{\boldsymbol{\beta}}(\lambda))< MSE(\widehat{\boldsymbol{\beta}})$ for all $\lambda>0$. For this reason, and taking into account the asymptotic behaviour of the VIF, the variables more appropriate to be raised will be the variables $\mathbf{X}_{3}$ and $\mathbf{X}_{5}$.
    \item The value of $\lim_{\lambda\rightarrow\infty} K( \widetilde{\mathbf{X}},\lambda)$ is lower than 20 only when variables  $\mathbf{X}_{3}$ and $\mathbf{X}_{5}$ are raised, and for this reason these variables could be preferred to be selected.
    \item Note that $MSE(\hat{\beta}(\lambda))$ is always less than the one obtained by OLS even for the variable $\mathbf{X}_{2}$ where $\lambda_{min}<1$ due to it is verified that $\lambda_{min} \leq 2 \cdot \frac{\widehat{\sigma}^{2}}{\widehat{\beta}_{2}^{2} \cdot \mathbf{e}_{2}^{t} \mathbf{e}_{2} - \widehat{\sigma}^{2}}$. The lowest MSE is obtained when variable $\mathbf{X}_{5}$ is raised.
\end{itemize}

However, with a methodological purpose Table \ref{resultsraiseexample1} shows the results of the raise estimation raising each variable with its $\lambda_{min}$:
\begin{itemize}
    \item It is also possible to check that there is not non-essential multicollinearity has been mitigated although initially it was not worrisome, since the coefficient of variation of the raised variables, CV($\tilde{X}_i$), has increased.
    \item It also obtained that the estimated variances of the estimated coefficients diminishes when raise regression is applied.
    \item Raise regression provides results in relation to the inference of the model which allows to conclude that estimated coefficient of variable $\mathbf{X}_4$ is not individually significant in any case, while estimated coefficient of variable $\mathbf{X}_3$ is individually significant when variables $\mathbf{X}_4$ and $\mathbf{X}_5$ are raised. The estimated coefficient of variable $\mathbf{X}_5$ is individually significant when variable $\mathbf{X}_3$ is raised. Then, the raised models presents more estimated coefficients indvidually significant than the OLS estimation.
\end{itemize}

Taking into account the values obtained for the VIFs, CN, MSE and the improvement obtained in the individual significant test, the variable $\mathbf{X}_{5}$ should be selected to be raised. However, over again with a illustrative purpose, the variables $\mathbf{X}_{3}$ and $\mathbf{X}_{5}$ will be raised with two alternatives for the selection of $\lambda$ :
\begin{enumerate}[a)]
    \item Table \ref{resultsraiseexample1vifmenor10} presents the raise estimation by using the value of $\lambda$ that get that all VIFs will be lower than 10. Note that in this case we obtain MSE lower than the one of OLS and also lower than the ones obtained from ridge regression for $k=0.0077$ and $k=0.0015$ and for the one presented by \cite{kibria2020new} in Table 10 (MSE=2170,96), who uses this same example to illustrate a new Ridge-Type Estimator. However, the CN is higher than the established thresholds.
    \item Table \ref{resultsNCexample1} presents the raise estimation raising variable $\mathbf{X}_3$ and $\mathbf{X}_5$ by using the first value of $\lambda$ that get that the CN will be lower than 20. Note that in both cases, the VIFs are lower than 4. In relation to MSE, the model obtained when raising variable $\mathbf{X}_3$ presents a MSE higher than the one obtained for the value of $\lambda$ that makes the VIF lower than 10 while the  model obtained when raising variable $\mathbf{X}_5$ presents a MSE lower than the one obtained for the value of $\lambda$ than makes the VIF lower than 10. The conclusions obtained from these results will be very similar to the one obtained from results presented in Table \ref{resultsraiseexample1}.

\end{enumerate}

Finally, Table \ref{resultsorthogonalexample1} shows the estimation of the residualization when variables $\mathbf{X}_3$ and $\mathbf{X}_5$ are residualized. Note that the value of the coefficient of determination is the same than in OLS and raise regression while  the MSE is much lower than the one obtained from OLS, ridge or raise regression. In all cases, the values for the VIFs and the CN are lower than the established thresholds. When variable $\mathbf{X}_3$ is residualized, the estimated coefficient of variable $\mathbf{X}_4$ becomes individually significant for a 90\% level of confidence. However, when variable $\mathbf{X}_5$ is residualized, the estimated coefficients of variables $\mathbf{X}_4$ and $\mathbf{X}_5$ are not individually significant.
Note also that the coefficient of variation (CV) is very high (from a theoretical point of view should be infinite) and, consequently, the non-essential multicollinearity is being mitigated. However, despite of these results and following \cite{Garcia2019} and \cite{GarciaSalmeron_York} it is relevant to highlight that residualization should be applied only if the residualized variable is interpretable. Thus,
\begin{itemize}
    \item If the variable $X_3$ is residualized, its estimated coefficient is interpreted as the part of tricalcium silicate not related to the tricalcium aluminate, the tetracalcium aluminoferrite and the $\beta$-dicalcium silicate.
    \item If the variable $X_5$ is residualized, its estimated coefficient is interpreted as the part  $\beta$-dicalcium silicate not related to the tricalcium aluminate, the tricalcium silicate and the tetracalcium aluminoferrite.
\end{itemize}
In this case, it can be a very complicated task for non-experts in the field.

\begin{table}
\begin{center}
\begin{tabular}{ccccccccc}
  \hline
  Raised&Variable 2& &Variable 3& &Variable 4& &Variable 5&\\
  &($\lambda$=0.230547)&VIF&($\lambda$=2.01277)&VIF&($\lambda$=54.8438)&VIF&($\lambda$=24.2248)&VIF\\
  \hline
  $\hat{\beta}_1(\lambda)$&88.8665& &95.3277& &71.482793& &48.7570& \\
                  &(57.8629)&&(23.6760)***&&(15.0332)***&&(4.9829)***&\\
  $\hat{\beta}_2(\lambda)$ &1.2605&25.762&  1.2176&7.514& 1.4537&1.078 &1.6901&3.306 \\
                  &(0.6052)**&&(0.3290)***&&(0.1246)***&&(0.2182)***&\\
  $\hat{\beta}_3(\lambda)$ &0.2416 &176.275&0.1693 &28.92 & 0.4177&18.856 & 0.6510&1.462\\
                 &(0.6025)&&(0.2402)&&(0.197041)**&&(0.0548)***&\\
  $\hat{\beta}_4(\lambda)$ &-0.1885 &31.393& -0.2401&8.242 & 0.001825&1.015 & 0.2441&3.211\\
                 &(0.6177)&&(0.3165)&&(0.013515)&&(0.1975)&\\
  $\hat{\beta}_5(\lambda)$&-0.4088 &194.673 & -0.4773&32.176 & -0.2348&19.025 &-0.0057&1.442\\
                  &(0.5886)&  &(0.2393)**&&(0.1840)&&(0.0281)&\\
                  \hline
  $R^2$& 0.9824 &&0.9824&&0.9824&&0.9824&\\
  $F_{statistic}$& 111.5*** &&111.5***&&111.5***&&111.5***&\\
  $MSE$&4050.054& &1645.21&&308.5229&&211.2751&\\
  $CN$&206.301&&83.897&&52.119&&16.97&\\
    CV($\tilde{X}_i$)&0.79361&&0.32824&&4.4718&&1.0056&\\
  \hline
\end{tabular}
\caption{Raise estimation with the $\lambda_{min}$ of each variable for Portland cement dataset}\label{resultsraiseexample1}
\end{center}
\end{table}

\begin{table}
\begin{center}
\begin{tabular}{ccccc}
  \hline
  Raised& Variable 3& &Variable 5&\\
  &($\lambda=4.65$)&VIF&($\lambda=4.595$)&VIF\\
  \hline
  $\hat{\beta}_1(\lambda)$ & 102.9625 && 50.7337 &\\
                  &(13.2350)***&&(13.16963)***&\\
  $\hat{\beta}_2(\lambda)$  & 1.1402 &4.768874&1.67& 4.3769\\
                  &(0.2621)***&&(0.2511)***& \\
  $\hat{\beta}_3(\lambda)$  & 0.0903 & 8.9387&0.6307&9.1570\\
                 &(0.1281)& &(0.1373)***&\\
  $\hat{\beta}_4(\lambda)$ & -0.3194& 4.8194& 0.22355& 4.5389\\
                 &(0.2420)&&(0.23486)& \\
  $\hat{\beta}_5(\lambda)$  & -0.5545& 9.9939&-0.02575&9.9928\\
                  &(0.1334)***& &(0.12673)&\\
  \hline
  $R^2$&0.9824&&0.9824&\\
  $F_{statistic}$& 111.5*** &&111.5***&\\
  MSE&1820.895&&309.8762&\\
  CN&46.358&&46.395&\\
  CV($\tilde{X}_i$)&0.34222&&0.58710&\\
    \hline
\end{tabular}
\caption{Raise estimation raising variable $\mathbf{X}_{3}$ and variable $\mathbf{X}_{5}$ with the first value of $\lambda$ that makes all VIFs lower than 10}\label{resultsraiseexample1vifmenor10}
\end{center}
\end{table}

 \begin{table}
\begin{center}
\begin{tabular}{ccccc}
  \hline
  Raised & Variable 3&& Variable 5 &\\
     &  ($\lambda=19.46$) & VIF & ($\lambda=15.93$) &  VIF\\
  \hline
  $\hat{\beta}_1(\lambda)$  & 109.27585 && 49.033075&\\
                  &(5.8073)***&&(5.8489)***&\\
  $\hat{\beta}_2(\lambda)$  & 1.07626 &3.761&1.687338& 3.374\\
                  &(0.2328)***& &(0.22049)***&\\
  $\hat{\beta}_3(\lambda)$   & 0.02493 &1.605& 0.648247&1.948\\
                 &(0.03538)& &(0.0063325)***&\\
  $\hat{\beta}_4(\lambda)$ & -0.38502 &3.563& 0.241269& 3.295\\
                 &(0.20810)& &(0.2001)&\\
  $\hat{\beta}_5(\lambda)$ & -0.618428 &1.853&-0.008509&1.982\\
                  &(0.05743)***&&(0.041881)&\\
                  \hline
  $R^2$&0.9824&&0.9824&\\
  $F_{statistic}$& 111.5*** &&111.5***&\\
  MSE&2231.592&&213.198&\\
  CN&19.9942&&19.99037&\\
  CV($\tilde{X}_i$)&0.5252&&0.79122&\\
  \hline
\end{tabular}
\caption{Raise estimation raising variable $\mathbf{X}_{3}$ and $\mathbf{X}_{5}$ with the first value of $\lambda$ that makes NC lower than 20}\label{resultsNCexample1}
\end{center}
\end{table}


\begin{table}
\begin{center}
\begin{tabular}{ccccc}
  \hline
  Residualized & Variable 3& VIF & Variable 5& VIF \\
  \hline
  $\widehat{\boldsymbol{\gamma_1}}$ & 111.684 && 48.1936&\\
                  &(4.6956)***&&(4.14)***&\\
  $\widehat{\boldsymbol{\gamma_2}}$  & 1.05185 &3.678&1.6959& 3.251\\
                  &(0.2302)***& &(0.2164)***&\\
  $\widehat{\boldsymbol{\gamma_3}}$  & 0.5101 &1& 0.6569&1.064\\
                 &(0.72379)& &(0.0468)***&\\
  $\widehat{\boldsymbol{\gamma_4}}$ & -0.41 &3.46& 0.25&3.142\\
                 &(0.2050)*& &(0.1954)&\\
  $\widehat{\boldsymbol{\gamma_5}}$  & -0.6428 &1.181&-0.1441&1\\
                  &(0.04584)***& &(0.7090)&\\
                  \hline
  $R^2$&0.9824&&0.9824&\\
  $F_{statistic}$& 111.5*** &&111.5***&\\
  MSE&22.66984&&17.72962&\\
  CN&15.305&&13.657&\\
  CV($\tilde{X}_i$)&3.56$\cdot 10^{15}$&&7.25$\cdot 10^{14}$&\\
  \hline
\end{tabular}
\caption{Residualization for Portland cement dataset residualizing variable $\mathbf{X}_{3}$ and $\mathbf{X}_{5}$}\label{resultsorthogonalexample1}
\end{center}
\end{table}

\subsection{Example 2: Euribor data}

The following regression model presents an analysis of the Euribor (\textbf{100\%}):
 \begin{equation}\label{Finanzas}
       \mathbf{Euribor}= \beta_{1}+\beta_{2}\mathbf{HICP}+\beta_{3}\mathbf{BC}+\beta_{4}\mathbf{GD}+\mathbf{u},
 \end{equation}
 in relation to the Harmonized Index of Consumer Prices (\textbf{100\%}) ($\mathbf{HICP}$), the Balance of Payments to net current account ($\mathbf{BC}$) in millions of euros and the Government Deficit to net non-finacial accounts ($\mathbf{GD}$) expressed in millions of euros. It is supposed that the random disturbance $\mathbf{u}$ is centered, homoscedastic and uncorrelated.  This dataset, previously used by \cite{Salmeron2020a} and \cite{Salmeron2020centered}, is composed by 47 Eurozone observations for the period January 2002 to July 2013 (quarterly and seasonally adjusted data).

 Table \ref{correxample2} presents the main measures in relation to the diagnose of multicollinearity: the correlation matrix of the predictors, the coefficients of variation, the VIFs obtained from OLS and the CN (obtained from the matrix $\mathbf{X}^{t}\mathbf{X}$ transformed to present unit length columns). Note that the coefficient of variation of the variable $\mathbf{HICP}$ is equal to 0.07033 which is lower than the threshold 0.1002506 provided by \cite{Salmeron2020a} to consider that the variable shows such a slight variability that it is highly related to the intercept provoking worrying  non-essential multicollinearity. The existence of this king o f worrying  multicollinearity is also suggested by the fact that the values of VIFs are lower than 10 while the CN is higher than 30. Remember that the VIF is not able to detect the non-essential multicollinearity contrary to what happens with the CN.

 \begin{table}
\begin{center}
\begin{tabular}{cccc}
  \hline
  &        $\mathbf{HICP}$&         $\mathbf{BC}$   &      $\mathbf{GD}$\\
$\mathbf{HICP}$ & 1 & 0.23129 &-0.4685 \\
$\mathbf{BC}$ & 0.23129 &1 &-0.07033 \\
 $\mathbf{GD}$ &-0.4685 &-0.07033 & 1 \\
\hline
VIF &1.351&1.059&1.285\\
\hline
CV&0.07033&4.38723&-0.556102\\
CN&39.35375&&\\
\hline
\end{tabular}
\caption{Main measures in relation to the diagnose of multicollinearity for Euribor dataset}\label{correxample2}
\end{center}
\end{table}

The traditional solution for the non-essential multicollinearity is to center the variable which was applied by \cite{Salmeron2020a}.
Alternatively, it will be possible to apply the raise regression treating to minimizes the MSE and/or mitigate multicollinearity raising the variable that provoke it (in this case $\mathbf{HICP}$). Table \ref{resultsexample2} presents the estimation by OLS (second column), the estimation raising the variable $\mathbf{HICP}$ with the first value of $\lambda$ that makes the CV higher than $0.1002506$ (third column), the first value of $\lambda$ that provides a CN lower than 20  (fourth column). the value of $\lambda$ that minimizes the mean square error (fifth column) and the value of $\lambda$ that stabilizes the norm of the estimators (sixth column). Figure \ref{norma} shows the relation between $\left\| \widehat{\boldsymbol{\beta}}(\lambda) \right\|$ and $\lambda$ (varying between 0 and 40, with jumps of 0.1) and Figure \ref{diferencianorma} shows the difference between the norm of the estimators calculated for a given $\lambda$ and for its immediately preceding one. That difference is lower than 0.001 from $\lambda=12.2$.

Note that the different models are similar in relation to individual inference and in all cases the VIFs are close to 1. In relation to MSE, when raising the variable $\mathbf{HICP}$
 with the value of $\lambda$ that minimizes the MSE, the model provides a reduction of 95.7157\% in the MSE in comparison with the MSE obtained by OLS. This reduction is also relevant (84.36\%) when it is compared to the MSE, 0.4339, obtained by \cite{Salmeron2020a} where the same model is estimated but centering variable $\mathbf{HICP}$. Note that the condition number is lower 20 under the criteria of minimising the MSE and the criteria of stabilizing the norm of the estimators. This last criteria requires a lower value for $\lambda$ and also presents a reduction of the MSE very relevant (95.4661\%).

\begin{figure}
  \centering
  \includegraphics[width=12cm]{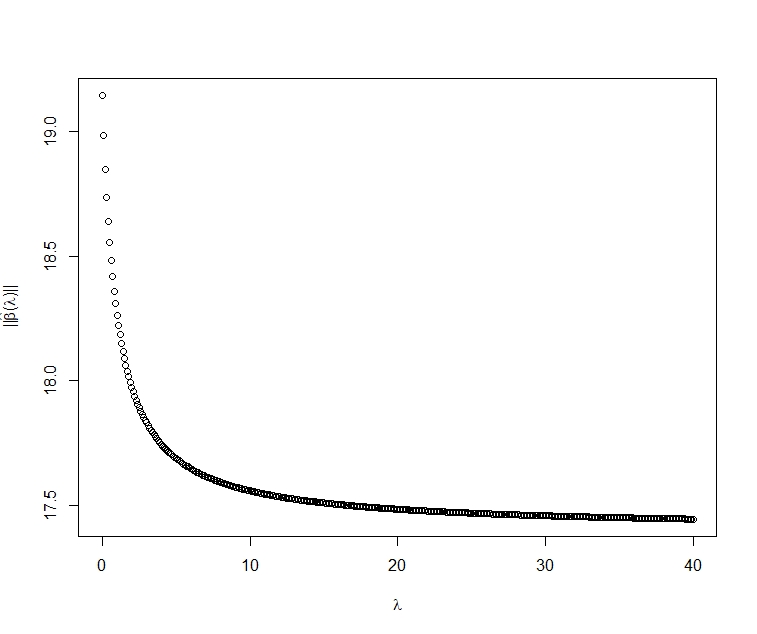}
  \caption{Relation between $\left\| \widehat{\boldsymbol{\beta}}(\lambda) \right\|$ and raising factor ($0<\lambda<40$).}\label{norma}
\end{figure}
\begin{figure}
  \centering
  \includegraphics[width=12cm]{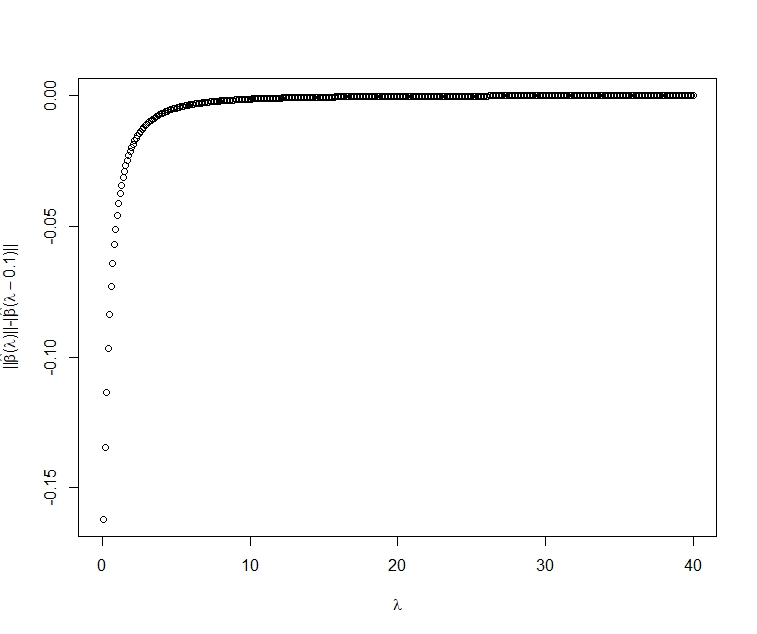}
  \caption{Behavior of the difference between $\left\| \widehat{\boldsymbol{\beta}}(\lambda) \right\|$  and $\left\| \widehat{\boldsymbol{\beta}}(\lambda - 0.1) \right\|$ in relation to raising factor  ($0<\lambda<40$).  }\label{diferencianorma}
\end{figure}

\begin{sidewaystable}
\small{\centering
\begin{tabular}{ccccccccccc}
  \hline
   & OLS&VIF& RAISE ($\lambda=0.55$)&VIF& RAISE ($\lambda=0.981$) &VIF & RAISE ($\lambda=37.3703$)&VIF&RAISE ($\lambda=12.2$)&VIF\\
  \hline
  $\hat{\beta}_1$ & 4.376 &  &                & & &&&&& \\
                  &(1.258)***& &           &   &&&&&&\\
  $\hat{\beta}_1(\lambda)$ & &  &  4.158                & & 4.275&&4.177&&4.1870& \\
                          && & (0.1837)***          &   &(0.6509)***&&(0.1683)*** &&(0.1902)***&\\
  $\hat{\beta}_2$  & -0.002  & 1.35&  &           &&& &&& \\
                   &(0.0126)& &                &&&&& &&\\
  $\hat{\beta}_2(\lambda)$  &  & & -0.002 &1.146           &- 0.0010&1.089&-5.38$\cdot 10^{-5}$ &1&-0.0001563&1.002 \\
                  && &        (0.013)        &&(0.00637)&&(3.289$\cdot 10^{-4}$)& &(0.000956)&\\
  $\hat{\beta}_3$ & -3.647$\cdot 10^{-5}$&1.059& &                & & &&&&   \\
                 &(3.897 $\cdot 10^{-6}$)***& &                & & &&&&&   \\
  $\hat{\beta}_3(\lambda)$ & &&  -3.647$\cdot 10^{-5}$&1.027&  -3.655$\cdot 10^{-5}$&1.019&-3.661$\cdot 10^{-5}$&1.005&-3.6607$\cdot 10^{-5}$&1.005 \\
                 &&&(3.897 $\cdot 10^{-6}$)***  &&(-3.823$\cdot 10^{-5}$)***&&(3.797$\cdot 10^{-6}$)***& &(-3.797$\cdot 10^{-6}$)*** \\
  $\hat{\beta}_4$  & 1.971 $\cdot 10^{-5}$ &1.284&& &                & & &&&  \\
                 &(2.202$\cdot 10^{-6}$)***&& &                & & &&&&   \\
  $\hat{\beta}_4(\lambda)$  &  &&1.971 $\cdot 10^{-5}$&1.121  &1.980$\cdot 10^{-5}$&1.076&1.988$\cdot 10^{-5}$&1.005&1.987$\cdot 10^{-5}$&1.007 \\
  &&& (2.202 $\cdot 10^{-6}$)*** &&(2.016$\cdot 10^{-6}$)***&&(1.948$\cdot 10^{-6}$)*** &&(1.950$\cdot 10^{-5}$)***&\\
  \hline
  $R^2$& 0.8314&&0.8314&&0.8314&&0.8314&&0.8314&\\
  $F_{statistic}$&70.68***&&70.68***&&70.68***&&70.68***&&70.68***&\\
  CN&39.35375&&25.471&&19.99&&4.2555&&4.909&\\
  MSE& 1.582988&&0.680108 &&0.4622&&0.06782&&0.07177&\\
  $\left\| \widehat{\boldsymbol{\beta}}(\lambda) \right\|$&19.14574&&18.51747&&18.27183&&17.4464&&17.5312&\\
  CV($\mathbf{HICP}$)&0.0703&&0.1004&&0.1251&&2.3232&&0.79991&\\
  \hline
\end{tabular}
\caption{OLS and raise estimations (with the first value of $\lambda$ that makes the CV higher than $0.1002506$, the $\lambda$ that makes $NC<20$, $\lambda_{min}$) and $\lambda$ that stabilizes the norm of the estimator for Euribor dataset} \label{resultsexample2}}
\end{sidewaystable}

\section{Conclusions}
    \label{conclusion}

Traditionally, penalized estimator (such as the ridge estimator, the Liu estimator, etc.) have been applied to estimate models with multicollinearity. These procedures exchange the mean square error by the bias and although the variance diminishes, the inference is lost and also the goodness of fit. Alternatively, the raise regression (\cite{Garcia2011} and \cite{Salmeron2017}) allows the mitigation of the problems generated by multicollinearity but without losing the inference and keeping the coefficient of determination.

This paper presents as novelty the generalization of  the raise estimator for any numbers of regressors together to its estimation, goodness of fit, global characteristics and its individual inference.
Anothers contributions are the development of the norm of the raised estimator, the analysis of the variance of the estimators, that the coefficient of variation increases when the raise regression is applied (and, as consequence, the raise regression is able to mitigate not only the essential multicollinearity but also the non-essential one) and complete some criteria to select the variable to be raised and the raising factor that can be applied individually or as a combination.
An interesting conclusion is the relation between the raise regression with the ordinary least squares and with the residualization.
Finally, the paper also present a comparison between the OLS, the raise and the ridge estimators under the criteria of MSE supporting the contribution of \cite{Garcia2020} who conclude that the ridge estimator is a particular case of the raising procedures.

Apart from these contributions, with the aim of having a complete vision of the technique, the paper also reviews some previous results in relation to the Mean Square Error, the Variance Inflation Factor, the Condition Number and the successive raise.

In conclusion, this paper contributes with the development of the raise regression with the goal of extend is application in many different fields where collinearity exists and ridge estimator is systematically applied without taking into consideration its lack of inference. Note that raise regression may mitigate essential and non-essential multicollinearity, maintaining the main characteristic of the original model and allowing the inference.

\section*{Acknowledgements}

This work has been supported by project PP2019-EI-02 of the University of Granada, Spain.

\bibliographystyle{Chicago}
\bibliography{bib}

\appendix

\section{Estimation of model (\ref{modelo0}) by OLS}
    \label{appendix_inicial}

   Given the general linear model (\ref{modelo0}) it is obtained that:
\begin{eqnarray}
    \widehat{\boldsymbol{\beta}} &=& \left( \mathbf{X}^{t} \mathbf{X} \right)^{-1} \mathbf{X}^{t} \mathbf{Y} = \left(
        \begin{array}{cc}
            \mathbf{X}_{-i}^{t} \mathbf{X}_{-i} & \mathbf{X}_{-i}^{t} \mathbf{X}_{i} \\
            \mathbf{X}_{i}^{t} \mathbf{X}_{-i} & \mathbf{X}_{i}^{t} \mathbf{X}_{i}
        \end{array} \right)^{-1} \cdot \left(
        \begin{array}{c}
            \mathbf{X}_{-i}^{t} \mathbf{Y} \\
            \mathbf{X}_{i}^{t} \mathbf{Y}
        \end{array} \right) \nonumber \\
        &=& \left(
        \begin{array}{cc}
            \mathbf{A} & \mathbf{B} \\
            \mathbf{B}^{t} & C
        \end{array} \right) \cdot \left(
        \begin{array}{c}
            \mathbf{X}_{-i}^{t} \mathbf{Y} \\
            \mathbf{X}_{i}^{t} \mathbf{Y}
        \end{array} \right) = \left(
        \begin{array}{c}
            \left( \mathbf{X}_{-i}^{t} \mathbf{X}_{-i} \right)^{-1} \mathbf{X}_{-i}^{t} \mathbf{Y} - \widehat{\boldsymbol{\alpha}} \cdot \frac{\mathbf{e}_{i}^{t} \mathbf{Y}}{\mathbf{e}_{i}^{t} \mathbf{e}_{i}}  \\
            \frac{\mathbf{e}_{i}^{t} \mathbf{Y}}{\mathbf{e}_{i}^{t} \mathbf{e}_{i}}
        \end{array} \right) = \left(
        \begin{array}{c}
            \widehat{\boldsymbol{\beta}}_{-i} \\
            \widehat{\beta}_{i}
        \end{array} \right), \label{eq_anexo_1}
\end{eqnarray}
only considering that $\mathbf{X} = [ \mathbf{X}_{i} \ \mathbf{X}_{-i} ]$, with $i=2,\dots,p$, and taking into account that:
\begin{eqnarray*}
    C &=& \left( \mathbf{X}_{i}^{t} \mathbf{X}_{i} - \mathbf{X}_{i}^{t} \mathbf{X}_{-i} \left( \mathbf{X}_{-i}^{t} \mathbf{X}_{-i} \right)^{-1} \mathbf{X}_{-i}^{t} \mathbf{X}_{i} \right)^{-1} = \left( \mathbf{X}_{i}^{t} \left( \mathbf{I} - \mathbf{X}_{-i} \left( \mathbf{X}_{-i}^{t} \mathbf{X}_{-i} \right)^{-1} \mathbf{X}_{-i}^{t} \right) \mathbf{X}_{i} \right)
    = \left( \mathbf{e}_{i}^{t} \mathbf{e}_{i} \right)^{-1}, \\
    \mathbf{B} &=& - \left( \mathbf{X}_{-i}^{t} \mathbf{X}_{-i} \right)^{-1} \mathbf{X}_{-i}^{t} \mathbf{X}_{i} \cdot \left( \mathbf{e}_{i}^{t} \mathbf{e}_{i} \right)^{-1}
    = - \widehat{\boldsymbol{\alpha}} \cdot \left( \mathbf{e}_{i}^{t} \mathbf{e}_{i} \right)^{-1} \\
    \mathbf{A} &=& \left( \mathbf{X}_{-i}^{t} \mathbf{X}_{-i} \right)^{-1} + \left( \mathbf{X}_{-i}^{t} \mathbf{X}_{-i} \right)^{-1} \mathbf{X}_{-i}^{t} \mathbf{X}_{i} \cdot \left( \mathbf{e}_{i}^{t} \mathbf{e}_{i} \right)^{-1} \mathbf{X}_{i}^{t} \mathbf{X}_{-i} \left( \mathbf{X}_{-i}^{t} \mathbf{X}_{-i} \right)^{-1}
    = \left( \mathbf{X}_{-i}^{t} \mathbf{X}_{-i} \right)^{-1} + \left( \mathbf{e}_{i}^{t} \mathbf{e}_{i} \right)^{-1} \cdot \widehat{\boldsymbol{\alpha}} \widehat{\boldsymbol{\alpha}}^{t},
\end{eqnarray*}
where $\widehat{\boldsymbol{\alpha}}$ y $\mathbf{e}_{i}^{t} \mathbf{e}_{i}$ are, respectively, the OLS estimator and the SSR of the auxiliary regression  (\ref{modelo_aux}) being $\mathbf{I}$ the identity matrix with adequate dimension.

On the other hand, supposing that the random disturbances are spherical, the individual inference will be given by the main diagonal of matrix $\left( \mathbf{X}^{t} \mathbf{X} \right)^{-1}$, it is to say:
\begin{equation}
    \label{eq_anexo_2}
    \left( \begin{array}{cc}
        \left( \mathbf{X}_{-i}^{t} \mathbf{X}_{-i} \right)^{-1} + \left( \mathbf{e}_{i}^{t} \mathbf{e}_{i} \right)^{-1} \cdot \widehat{\boldsymbol{\alpha}} \widehat{\boldsymbol{\alpha}}^{t} & - \widehat{\boldsymbol{\alpha}} \cdot \left( \mathbf{e}_{i}^{t} \mathbf{e}_{i} \right)^{-1} \\
        - \widehat{\boldsymbol{\alpha}}^{t} \cdot \left( \mathbf{e}_{i}^{t} \mathbf{e}_{i} \right)^{-1} & \left( \mathbf{e}_{i}^{t} \mathbf{e}_{i} \right)^{-1}
    \end{array} \right).
\end{equation}

In this case, the null hypothesis $H_{0}: \beta_{i} = 0$ is rejected (in favor of the alternative hypothesis $H_{1}: \beta_{i} \not= 0$) if:
\begin{equation}
    \label{inferencia_mco1}
    t_{exp} (\beta_{i}) = \left| \frac{\widehat{\beta}_{i}}{\widehat{\sigma} \cdot \sqrt{\left( \mathbf{e}_{i}^{t} \mathbf{e}_{i} \right)^{-1}}} \right| = \left| \frac{\mathbf{e}_{i}^{t} \mathbf{Y}}{\widehat{\sigma} \sqrt{\mathbf{e}_{i}^{t} \mathbf{e}_{i}}} \right| > t_{n-p}(1-\alpha/2),
\end{equation}
where $\widehat{\sigma}^{2}$ is  the estimation by OLS of the variance of the random disturbance:
$$\widehat{\sigma}^{2} = \frac{\mathbf{e}^{t} \mathbf{e}}{n-p},$$
where $\mathbf{e}$ are the residuals of model (\ref{modelo0}).

At the same time, the null hypothesis $H_{0}: \beta_{-i,j} = 0$ is rejected (in favor of the alternative hypothesis  $H_{1}: \beta_{-i,j} \not= 0$) if:
\begin{equation}
    \label{inferencia_mco2}
    t_{exp} (\beta_{-i,j}) = \left| \frac{\widehat{\beta}_{-i,j}}{\widehat{\sigma} \cdot \sqrt{w_{jj} + \left( \mathbf{e}_{i}^{t} \mathbf{e}_{i} \right)^{-1} \cdot \widehat{\alpha}_{j}^{2}}} \right| > t_{n-p}(1-\alpha/2), \quad j=1,\dots,i-1,i+1,\dots,p.
\end{equation}
Note that the element $j$ of $\widehat{\beta}_{-i}$ is $\widehat{\beta}_{-i,j} = \widehat{\gamma}_{-i,j} - \widehat{\alpha}_{j} \cdot \widehat{\beta}_{i}$ where $\gamma_{-i,j}$ is the element $j$ of $\left( \mathbf{X}_{-i}^{t} \mathbf{X}_{-i} \right)^{-1} \mathbf{X}_{-i}^{t} \mathbf{Y}$.

Finally, the residuals of the model are given by the following expression:
\begin{eqnarray}
    \mathbf{e} &=& \mathbf{Y} - \widehat{\mathbf{Y}} = \mathbf{Y} - \left( \mathbf{X}_{-i} \ \mathbf{X}_{i} \right) \cdot \left(
        \begin{array}{c}
            \left( \mathbf{X}_{-i}^{t} \mathbf{X}_{-i} \right)^{-1} \mathbf{X}_{-i}^{t} \mathbf{Y} - \widehat{\boldsymbol{\alpha}} \cdot \frac{\mathbf{e}_{i}^{t} \mathbf{Y}}{\mathbf{e}_{i}^{t} \mathbf{e}_{i}}  \\
            \frac{\mathbf{e}_{i}^{t} \mathbf{Y}}{\mathbf{e}_{i}^{t} \mathbf{e}_{i}}
        \end{array} \right) \nonumber \\
        &=& \mathbf{Y} - \mathbf{X}_{-i} \left( \mathbf{X}_{-i}^{t} \mathbf{X}_{-i} \right)^{-1} \mathbf{X}_{-i}^{t} \mathbf{Y} + \mathbf{X}_{-i} \widehat{\boldsymbol{\alpha}} \cdot \frac{\mathbf{e}_{i}^{t} \mathbf{Y}}{\mathbf{e}_{i}^{t} \mathbf{e}_{i}} - \mathbf{X}_{i} \frac{\mathbf{e}_{i}^{t} \mathbf{Y}}{\mathbf{e}_{i}^{t} \mathbf{e}_{i}} 
        = \mathbf{Y} - \mathbf{X}_{-i} \left( \mathbf{X}_{-i}^{t} \mathbf{X}_{-i} \right)^{-1} \mathbf{X}_{-i}^{t} \mathbf{Y} - \mathbf{e}_{i} \frac{\mathbf{e}_{i}^{t} \mathbf{Y}}{\mathbf{e}_{i}^{t} \mathbf{e}_{i}}, \label{eq_anexo_3}
\end{eqnarray}
due to being $\mathbf{e}_{i}$ the residuals of the auxiliary regression (\ref{modelo_aux}) it is verified that $\mathbf{e}_{i} = \mathbf{X}_{i} - \mathbf{X}_{-i} \widehat{\boldsymbol{\alpha}}$.

\section{Calculation of the raise estimation}
    \label{appendixA}

    Taking into account the inversion of the partitioned matrices, it is obtained that:
    $$\widehat{\boldsymbol{\beta}}(\lambda) = \left(
            \begin{array}{cc}
                \mathbf{A}_{\lambda} & \mathbf{B}_{\lambda} \\
                \mathbf{B}_{\lambda}^{t} & C_{\lambda}
            \end{array} \right) \cdot \left(
            \begin{array}{c}
                \mathbf{X}_{-i}^{t} \mathbf{Y} \\
                \mathbf{X}_{-i}^{t} \mathbf{Y} + \lambda \mathbf{e}_{i}^{t} \mathbf{Y}
            \end{array} \right) = \left(
            \begin{array}{c}
                \mathbf{A}_{\lambda} \cdot \mathbf{X}_{-i}^{t} \mathbf{Y} + \mathbf{B}_{\lambda} \cdot \mathbf{X}_{-i}^{t} \mathbf{Y} + \lambda \cdot \mathbf{B}_{\lambda} \cdot \mathbf{e}_{i}^{t} \mathbf{Y} \\
                \mathbf{B}_{\lambda}^{t} \mathbf{X}_{-i}^{t} \mathbf{Y} + C_{\lambda} \cdot \mathbf{X}_{-i}^{t} \mathbf{Y} + \lambda \cdot C_{\lambda} \cdot \mathbf{e}_{i}^{t} \mathbf{Y}
            \end{array} \right) = \left(
            \begin{array}{c}
                \widehat{\boldsymbol{\beta}}(\lambda)_{-i} \\
                \widehat{\beta}(\lambda)_{i}
            \end{array} \right),$$
    where:
    \begin{eqnarray}
        C_{\lambda} &=& \left( \mathbf{X}_{i}^{t} \mathbf{X}_{i} + (\lambda^{2} + 2\lambda) \cdot \mathbf{e}_{i}^{t}\mathbf{e}_{i} - \mathbf{X}_{i}^{t} \mathbf{X}_{-i} \left( \mathbf{X}_{-i}^{t} \mathbf{X}_{-i} \right)^{-1} \mathbf{X}_{i}^{t} \mathbf{X}_{i} \right)^{-1} \nonumber \\
            &=& \left(  \mathbf{X}_{i}^{t} \left( \mathbf{I} - \mathbf{X}_{-i} \left( \mathbf{X}_{-i}^{t} \mathbf{X}_{-i} \right)^{-1} \mathbf{X}_{-i}^{t} \right)  \mathbf{X}_{i} + (\lambda^{2} + 2\lambda) \cdot \mathbf{e}_{i}^{t}\mathbf{e}_{i} \right)^{-1} \nonumber \\
            &=& \left( \mathbf{e}_{i}^{t}\mathbf{e}_{i} + (\lambda^{2} + 2\lambda) \cdot \mathbf{e}_{i}^{t}\mathbf{e}_{i} \right)^{-1}
            = \left( (\lambda^{2} + 1)^{2} \cdot \mathbf{e}_{i}^{t}\mathbf{e}_{i} \right)^{-1}, \nonumber \\
        \mathbf{B}_{\lambda} &=& -  \left( \mathbf{X}_{-i}^{t} \mathbf{X}_{-i} \right)^{-1} \mathbf{X}_{-i}^{t} \mathbf{X}_{i} \cdot C_{\lambda} = - \left( (\lambda^{2} + 1)^{2} \cdot \mathbf{e}_{i}^{t}\mathbf{e}_{i} \right)^{-1} \cdot \widehat{\boldsymbol{\alpha}}, \nonumber \\
        \mathbf{A}_{\lambda} &=& \left( \mathbf{X}_{-i}^{t} \mathbf{X}_{-i} \right)^{-1} + \left( \mathbf{X}_{-i}^{t} \mathbf{X}_{-i} \right)^{-1}  \mathbf{X}_{-i}^{t} \mathbf{X}_{i} \cdot C_{\lambda} \cdot  \mathbf{X}_{i}^{t} \mathbf{X}_{-i} \left( \mathbf{X}_{-i}^{t} \mathbf{X}_{-i} \right)^{-1} \nonumber \\
            &=& \left( \mathbf{X}_{-i}^{t} \mathbf{X}_{-i} \right)^{-1} + \left( (\lambda^{2} + 1)^{2} \cdot \mathbf{e}_{i}^{t}\mathbf{e}_{i} \right)^{-1} \cdot \widehat{\boldsymbol{\alpha}} \widehat{\boldsymbol{\alpha}}^{t}. \nonumber
    \end{eqnarray}

    Then:
    \begin{eqnarray}
        \widehat{\boldsymbol{\beta}}(\lambda)_{-i} &=& \left( \mathbf{X}_{-i}^{t} \mathbf{X}_{-i} \right)^{-1} \mathbf{X}_{-i}^{t} \mathbf{Y} +  \left( (\lambda^{2} + 1)^{2} \cdot \mathbf{e}_{i}^{t}\mathbf{e}_{i} \right)^{-1} \cdot \widehat{\boldsymbol{\alpha}} \widehat{\boldsymbol{\alpha}}^{t} \cdot \mathbf{X}_{-i}^{t} \mathbf{Y} - \left( (\lambda^{2} + 1)^{2} \cdot \mathbf{e}_{i}^{t}\mathbf{e}_{i} \right)^{-1} \cdot \widehat{\boldsymbol{\alpha}} \mathbf{X}_{i}^{t} \mathbf{Y} \nonumber \\
            & & - \lambda \left( (\lambda^{2} + 1)^{2} \cdot \mathbf{e}_{i}^{t}\mathbf{e}_{i} \right)^{-1} \cdot \widehat{\boldsymbol{\alpha}} \mathbf{e}_{i}^{t} \mathbf{Y} \nonumber \\
            &=& \left( \mathbf{X}_{-i}^{t} \mathbf{X}_{-i} \right)^{-1} \mathbf{X}_{-i}^{t} \mathbf{Y} +  \left( (\lambda^{2} + 1)^{2} \cdot \mathbf{e}_{i}^{t}\mathbf{e}_{i} \right)^{-1} \cdot \widehat{\boldsymbol{\alpha}} \cdot \left( \widehat{\boldsymbol{\alpha}}^{t} \mathbf{X}_{-i}^{t} - \mathbf{X}_{i}^{t} - \lambda \mathbf{e}_{i}^{t} \right) \cdot \mathbf{Y} \nonumber \\
            &=&  \left( \mathbf{X}_{-i}^{t} \mathbf{X}_{-i} \right)^{-1} \mathbf{X}_{-i}^{t} \mathbf{Y} +  \left( (\lambda^{2} + 1)^{2} \cdot \mathbf{e}_{i}^{t}\mathbf{e}_{i} \right)^{-1} (1 + \lambda) \cdot \widehat{\boldsymbol{\alpha}} \cdot \mathbf{e}_{i}^{t} \mathbf{Y} \nonumber \\
            &=& \left( \mathbf{X}_{-i}^{t} \mathbf{X}_{-i} \right)^{-1} \mathbf{X}_{-i}^{t} \mathbf{Y} - \widehat{\boldsymbol{\alpha}} \cdot \frac{\mathbf{e}_{i}^{t} \mathbf{Y}}{(1+\lambda) \cdot \mathbf{e}_{i}^{t} \mathbf{e}_{i}}, \nonumber \\
        \widehat{\beta}(\lambda)_{i} &=& - \left( (\lambda^{2} + 1)^{2} \cdot \mathbf{e}_{i}^{t}\mathbf{e}_{i} \right)^{-1} \cdot \left( \widehat{\boldsymbol{\alpha}}^{t} \mathbf{X}_{-i}^{t} - \mathbf{X}_{i}^{t} - \lambda \mathbf{e}_{i}^{t} \right) \cdot \mathbf{Y}
            = \left( (\lambda^{2} + 1)^{2} \cdot \mathbf{e}_{i}^{t}\mathbf{e}_{i} \right)^{-1} \cdot (1 + \lambda) \mathbf{e}_{i}^{t} \mathbf{Y} \nonumber \\
            &=& \frac{\mathbf{e}_{i}^{t} \mathbf{Y}}{(1+\lambda) \cdot \mathbf{e}_{i}^{t} \mathbf{e}_{i}}, \nonumber
    \end{eqnarray}
    where it was used that $\widehat{\boldsymbol{\alpha}}^{t} \mathbf{X}_{-i}^{t} - \mathbf{X}_{i}^{t} = - \mathbf{e}_{i}^{t}$ for $i=2,\dots,p$.

\section{Residualization}
    \label{appendix_ortogonal}

    \cite{Garcia2019} formalizes the residualization regression. Thus, parting from model (\ref{modelo0}), the following expression is obtained for the regression with orthogonal variables:
    \begin{equation}
        \label{modelo_ortogonal}
        \mathbf{Y} = \mathbf{X}_{O} \boldsymbol{\gamma} + \mathbf{w},
    \end{equation}
    where $\mathbf{X}_{O} = \left( \mathbf{X}_{-i} \ \mathbf{e}_{i} \right)$, being $\mathbf{X}_{-i}$ the result of eliminating the column (variable) $i$ of matrix $\mathbf{X}$, $i=2,\dots,p$,  and being $\mathbf{e}_{i}$ the residuals of the auxiliary regression (\ref{modelo_aux}).

Then, the estimation by OLS of model (\ref{modelo_ortogonal}) is given by:
\begin{equation}
    \widehat{\boldsymbol{\gamma}} = \left( \mathbf{X}_{O}^{t} \mathbf{X}_{O} \right)^{-1} \mathbf{X}_{O}^{t} \mathbf{Y} =
         \left(
        \begin{array}{cc}
            \left( \mathbf{X}_{-i}^{t} \mathbf{X}_{-i}  \right)^{-1} & \mathbf{0} \\
            \mathbf{0} & \left( \mathbf{e}_{i}^{t} \mathbf{e}_{i} \right)^{-1}
        \end{array} \right) \cdot \left(
        \begin{array}{c}
            \mathbf{X}_{-i}^{t} \mathbf{Y} \\
            \mathbf{e}_{i}^{t} \mathbf{Y}
        \end{array} \right) = \left(
        \begin{array}{c}
            \left( \mathbf{X}_{-i}^{t} \mathbf{X}_{-i}  \right)^{-1} \mathbf{X}_{-i}^{t} \mathbf{Y} \\
            \frac{\mathbf{e}_{i}^{t} \mathbf{Y}}{\mathbf{e}_{i}^{t} \mathbf{e}_{i}}
        \end{array} \right) = \left(
        \begin{array}{c}
            \widehat{\boldsymbol{\gamma}}_{-i} \\
            \widehat{\gamma}_{i}
        \end{array} \right). \label{est_ortogonal}
\end{equation}
Note, see expression (\ref{eq_anexo_1}), that $\widehat{\boldsymbol{\gamma}}_{i} = \widehat{\boldsymbol{\beta}}_{i}$.

On the other hand, the residuals of model (\ref{modelo_ortogonal}) are given by:
\begin{eqnarray}
    \mathbf{e}_{0} &=& \mathbf{Y} - \widehat{\mathbf{Y}}_{O} = \mathbf{Y} - \left( \mathbf{X}_{-i} \ \mathbf{e}_{i} \right) \cdot \left(
        \begin{array}{c}
            \left( \mathbf{X}_{-i}^{t} \mathbf{X}_{-i} \right)^{-1} \mathbf{X}_{-i}^{t} \mathbf{Y} \\
            \frac{\mathbf{e}_{i}^{t} \mathbf{Y}}{\mathbf{e}_{i}^{t} \mathbf{e}_{i}}
        \end{array} \right) \nonumber \\
        &=& \mathbf{Y} - \mathbf{X}_{-i} \left( \mathbf{X}_{-i}^{t} \mathbf{X}_{-i} \right)^{-1} \mathbf{X}_{-i}^{t} \mathbf{Y} - \mathbf{e}_{i} \frac{\mathbf{e}_{i}^{t} \mathbf{Y}}{\mathbf{e}_{i}^{t} \mathbf{e}_{i}}. \label{eq5}
\end{eqnarray}

It can be observed, see expressions (\ref{residuos_alzado}) and (\ref{eq_anexo_3}), that these expressions coincide with the ones of models (\ref{modelo0}) and (\ref{modelo_alzado}). Consequently, it is also verified that $\widehat{\sigma}^{2}_{O} = \widehat{\sigma}^{2}$.

Taking into account the expression (\ref{eq_anexo_2}), is evident that the main diagonal of matrix $\left( \mathbf{X}_{O}^{t} \mathbf{X}_{O} \right)^{-1}$ differs to the one associated to model (\ref{modelo0}) except for the element $i$. In this case, the null hypothesis  $H_{0}: \gamma_{i} = 0$ is rejected (in favor of the alternative hypothesis $H_{1}: \gamma_{i} \not= 0$) if:
\begin{equation}
    \label{inferencia_ortogonal1}
    t_{exp} (\gamma_{i}) = \left| \frac{\widehat{\gamma}_{i}}{\widehat{\sigma} \cdot \sqrt{\left( \mathbf{e}_{i}^{t} \mathbf{e}_{i} \right)^{-1}}} \right|  = \left| \frac{\mathbf{e}_{i}^{t} \mathbf{Y}}{\widehat{\sigma} \sqrt{\mathbf{e}_{i}^{t} \mathbf{e}_{i}}} \right|  > t_{n-p}(1-\alpha/2).
\end{equation}

Note that from (\ref{inferencia_mco1}) it is possible to establish that $t_{exp} (\gamma_{i}) = t_{exp} (\beta_{i})$.

At the same time, the null hypothesis $H_{0}: \gamma_{-i,j} = 0$ is rejected (in favor of the alternative hypothesis $H_{1}: \gamma_{-i,j} \not= 0$) if:
\begin{equation}
    \label{inferencia_ortogonal2}
    t_{exp} (\gamma_{-i,j}) = \left| \frac{\boldsymbol{\gamma}_{-i,j}}{\widehat{\sigma} \cdot \sqrt{w_{jj}}} \right| > t_{n-p}(1-\alpha/2). \quad j=1,\dots,i-1,i+1,\dots,p.
\end{equation}
Note that $\boldsymbol{\gamma}_{-i,j}$ is the element $j$ of $\boldsymbol{\gamma}_{-i}$.

From the point of view of multicollinearity, the VIF of the orthogonalized variable is equal to 1 while the VIFs of the rest of variable coincide to the VIF of matrix  $\mathbf{X}_{-i}$. Analogously, the CN obtained from the matrix   $\mathbf{X}_{O}$ coincide to the one obtained from matrix $\mathbf{X}_{-i}$.

Finally, the MSE of the residualization method is:
$$MSE \left( \boldsymbol{\gamma} \right) = \sigma^{2} \cdot \left[ tr \left( (\mathbf{X}_{-i}^{t} \mathbf{X}_{-i})^{-1} \right)  + \left( \mathbf{e}_{i}^{t} \mathbf{e}_{i} \right)^{-1} \right] + \beta_{i}^{2} \cdot \sum \limits_{j=0, j \not= i}^{p} \widehat{\alpha}_{j}^{2}.$$

\section{Behavior of the estimators in MSE}
    \label{appendix_mse}

It is known that, see \cite{Theobald1974},  given the general linear model %
\eqref{modelo0}, it is possible to compare two estimators with
the form $\boldsymbol{\hat{\beta}}_{i}=\mathbf{C}_{i}\mathbf{Y}$ with $i=1,2$
by applying the criterion of the matrix of the root mean square (MtxMSE)
defined by:
\begin{equation}\label{eq: Matriz de error cuadratico medio}
\textup{MtxMSE}\left( \widehat{\boldsymbol{\beta }}\right) =E\left[ \left(
\widehat{\boldsymbol{\beta }}-\boldsymbol{\beta }\right) \left( \widehat{%
\boldsymbol{\beta }}-\boldsymbol{\beta }\right) ^{t}\right] .
\end{equation}

From:
\begin{equation}\label{eq: delta}
\boldsymbol{\Delta }=\textup{MtxMSE}\left( \widehat{\boldsymbol{\beta }}%
_{2}\right) -\textup{MtxMSE}\left( \widehat{\boldsymbol{\beta }}_{1}\right)
,
\end{equation}%
it is known that if the matrix $\boldsymbol{\Delta }$ is definite or
positive semidefinite, then estimator $\widehat{\boldsymbol{\beta }}_{1}$
will be more appropriate than estimator $\widehat{\boldsymbol{\beta }}_{2}$.
It is also verified that $\theta =MSE\left( \widehat{\boldsymbol{\beta }}%
_{2}\right) -MSE\left( \widehat{\boldsymbol{\beta }}_{1}\right) \geq 0$.
Hence, it is possible to conclude that $\widehat{\boldsymbol{\beta }}_{1}$
is better than $\widehat{\boldsymbol{\beta }}_{2}$ in terms of MSE.

Based on \cite{Farebrother1976}, it is possible to rewrite \eqref{eq: delta}
as:
\begin{equation}
\boldsymbol{\Delta }=\sigma ^{2}\mathbf{S}-\textup{Bias}\left( \widehat{%
\boldsymbol{\beta }}_{1}\right) \textup{Bias}\left( \widehat{\boldsymbol{%
\beta }}_{2}\right) ,  \label{eq: delta 2}
\end{equation}%
where $\mathbf{S} = \mathbf{C}_{2} \mathbf{C}_{2}^{t} - \mathbf{C}_{1} \mathbf{C}_{1}^{t}$.

From the above expression and following Theorem 2.2.2 of \cite{Trenklar1980}, it is possible to obtain the following result.

\begin{proposition}
\label{Proposition 1} Let $\widehat{\boldsymbol{\beta }}_{i}=\mathbf{C}_{i}%
\mathbf{Y}$, with $i=1,2$, be two homogeneous linear estimators of $\boldsymbol{\beta}$ such that $\mathbf{S}$ is a positive definite matrix.
Furthermore, let the following inequality be valid:
\begin{equation}
\boldsymbol{\beta }^{t}\left( \mathbf{C}_{1}\mathbf{X}-\mathbf{I}%
_{p}\right) ^{t}\mathbf{S}^{-1}\left( \mathbf{C}_{1}\mathbf{X}-\mathbf{%
I}_{p}\right) \boldsymbol{\beta }<\sigma ^{2},  \label{eq: inequality}
\end{equation}%
then $\boldsymbol{\Delta }=\textup{MtxMSE}\left( \widehat{\boldsymbol{\beta
}}_{2}\right) -\textup{MtxMSE}\left( \widehat{\boldsymbol{\beta }}%
_{1}\right) >0$.
\end{proposition}

If it is verified that $\mathbf{S}$ is a positive definite matrix, based on
Proposition \ref{Proposition 1} it is possible to state that estimator $%
\widehat{\boldsymbol{\beta }}_{1}$ is better than estimator $\widehat{%
\boldsymbol{\beta }}_{2}$ under the root mean square error matrix criterion
and the MSE criterion (see \cite{Theobald1974}, \cite{Farebrother1976} and \cite{Trenklar1980}).

\end{document}